\documentclass[10pt]{article}

\usepackage{fullpage}
\usepackage{setspace}
\usepackage{parskip}
\usepackage{titlesec}
\usepackage[section]{placeins}
\usepackage{xcolor}
\usepackage{breakcites}
\usepackage{lineno}
\usepackage{hyphenat}
\usepackage{subfig}
\usepackage{amsmath}
\usepackage{amsthm}
\usepackage{algorithm}
\usepackage{algpseudocode}

\usepackage[numbers]{natbib}
\setcitestyle{square}
\setcitestyle{comma}

\newtheorem*{theorem*}{Theorem}

\newtheorem{lemma}{Lemma}

\PassOptionsToPackage{hyphens}{url}
\usepackage[colorlinks = true,
            linkcolor = blue,
            urlcolor  = blue,
            citecolor = blue,
            anchorcolor = blue]{hyperref}
\usepackage{etoolbox}

\renewenvironment{abstract}
  {{\bfseries\noindent{\abstractname}\par\nobreak}\footnotesize}
  {\bigskip}

\titlespacing{\section}{0pt}{*3}{*1}
\titlespacing{\subsection}{0pt}{*2}{*0.5}
\titlespacing{\subsubsection}{0pt}{*1.5}{0pt}

\usepackage{authblk}

\usepackage{graphicx}
\usepackage[space]{grffile}
\usepackage{latexsym}
\usepackage{textcomp}
\usepackage{longtable}
\usepackage{tabulary}
\usepackage{booktabs,array,multirow}
\usepackage{amsfonts,amsmath,amssymb, bbm}
\providecommand\citet{\cite}
\providecommand\citep{\cite}

\newif\iflatexml\latexmlfalse
\newcommand{\RNum}[1]{\textup{\uppercase\expandafter{\romannumeral#1}}}

\AtBeginDocument{\DeclareGraphicsExtensions{.pdf,.PDF,.eps,.EPS,.png,.PNG,.tif,.TIF,.jpg,.JPG,.jpeg,.JPEG}}

\usepackage[utf8]{inputenc}
\usepackage[english]{babel}

\usepackage{float}

\begin{document}

\title{Multi-Agent Reinforcement Learning for Network Routing in Integrated Access Backhaul Networks}

\author[1]{Shahaf Yamin, Haim H. Permuter}%
\affil[1]{The authors are with the School of Electrical and Computer Engineering, Ben-Gurion University of the Negev, Be'er-Sheva, Israel (e-mail: yamins@post.bgu.ac.il; haimp@bgu.ac.il).}%

\vspace{-1em}

  \date{\today}

\begingroup
\let\center\flushleft
\let\endcenter\endflushleft
\maketitle
\endgroup

\selectlanguage{english}
\begin{abstract}
In this study, we examine the problem of wireless routing in integrated access backhaul (IAB) networks involving fiber-connected base stations, wireless base stations, and multiple users. Physical constraints prevent the use of a central controller, leaving base stations with limited access to real-time network conditions. These networks operate in a time-slotted regime, where base stations monitor network conditions and forward packets accordingly.
Our objective is to maximize the arrival ratio of packets, while simultaneously minimizing their latency. To accomplish this, we formulate this problem as a multi-agent partially observed Markov Decision Process (POMDP). Moreover, we develop an algorithm that uses Multi-Agent Reinforcement Learning (MARL) combined with Advantage Actor Critic (A2C) to derive a joint routing policy on a distributed basis. Due to the importance of packet destinations for successful routing decisions, we utilize information about similar destinations as a basis for selecting specific-destination routing decisions. For portraying the similarity between those destinations, we rely on their relational base-station associations, i.e., which base station they are currently connected to. Therefore, the algorithm is referred to as Relational Advantage Actor Critic (Relational A2C). To the best of our knowledge, this is the first work that optimizes routing strategy for IAB networks. Further, we present three types of training paradigms for this algorithm in order to provide flexibility in terms of its performance and throughput. Through numerical experiments with different network scenarios, Relational A2C algorithms were demonstrated to be capable of achieving near-centralized performance even though they operate in a decentralized manner in the network of interest. Based on the results of those experiments, we compare Relational A2C to other reinforcement learning algorithms, like Q-Routing and Hybrid Routing. This comparison illustrates that solving the joint optimization problem increases network efficiency and reduces selfish agent behavior.\end{abstract}%

\sloppy
\section{Introduction}

{\label{874460}}

The increasing demand for wireless communication and the limitations of the electromagnetic spectrum have led to the development of more efficient methods for managing networks. To meet these needs, the 3rd Generation Partnership Project (3GPP) has established a standard, New Radio (NR), which includes novel designs and technologies to support fifth-generation (5G) networks \cite{5G_spec}. One of the key features of this new protocol is the inclusion of new bands at millimeter wave (mmWave) frequencies. These frequencies offer the potential for increased data rates by exploiting the spatial diversity available in these bands, which are currently less congested compared to traditional bands. However, operating at mmWave frequencies also introduces new physical challenges, such as severe path and penetration losses. To overcome these challenges, network density can be increased and beam-forming methods can be used \cite{van1988beamforming}.
\par
Although increasing network density has potential benefits, the deployment and operation of fiber between the Next-Generation Node Base Station (gNB) and the core network can be costly. Integrated Access and Backhaul (IAB) is a promising solution for successful 5G adoption as it allows for only a fraction of the gNBs to be connected to traditional fiber-like infrastructures, thus reducing redundant deployment and operational costs by utilizing spatial diversity \cite{5G_IAB}. The gNBs connected to fiber are called IAB donors, while the remaining gNBs are called IAB nodes and use multi-hop wireless connections for backhaul traffic. While IAB networks are cost-effective in terms of deployment and operation, ensuring reliable network performance remains a challenging research area due to their highly non-stationary nature. The dynamic nature of the topology, tight delay constraints, and limited information regarding the status of the network are some of the factors to be considered when supporting these requirements.
\par
Routing plays a crucial role in the context of network congestion control, where each destination may have multiple paths, and base stations monitor network conditions to make routing decisions. There are two main approaches for implementing routing algorithms in wireless networks: centralized and distributed. In a centralized approach, there is a central network processor that is responsible for path selection, while in a distributed approach, each node makes next-hop decisions based only on its own observations without knowledge of other nodes' decisions. In practical implementations, due to bandwidth limitations and multi-hop structures, information sharing is limited to the base station's neighborhood. This implies that base stations can only observe a part of the current network state, and enhance, when operating in a distributed manner next-hop transmission decisions are only based on partial observations.
\par
This paper focuses on the analysis of distributed routing algorithms developed for networks that exhibit physical limitations, specifically IAB-based networks. We design a deep reinforcement-learning-based algorithm to achieve optimal routing policies. One unique aspect of our approach is that, given that successful routing decisions rely heavily on packet destinations, we consider that it is beneficial for the agent to use knowledge of similar destinations when making routing decisions. To portray the similarity between these destinations, we use their \textit{relational} base-station associations, which refer to the base stations they are currently connected to. Our study differs from previous works that focused on designing routing policies to optimize packet paths to specific destinations without sharing information with other similar destinations. These previous algorithms are generally not suitable for a dynamic topology like the one in this study. We propose a novel algorithm for learning a decentralized routing policy using deep policy gradients. To make the learning efficient, we use the Advantage Actor Critic (A2C) \cite{actor_critic} algorithm, which is a combination of policy gradient \cite{reiforce_algorithm}, temporal difference estimation \cite{td_learning}, and deep neural networks, and enables learning from experience in an unknown environment with a large state space through interactions with the environment. Our proposed algorithm is called \textit{Relational A2C}.
\par
In the current study, we present numerical results that evaluate the performance of the Relational A2C algorithm in several network scenarios. We also compare the results obtained by Relational A2C with those from other methods such as \cite{boyan1994packet,actor_critic_for_adaptive_routing_hybrid_method,bellman1958routing,backpressure_routing}. Our results show that the proposed approach outperforms existing methods and is able to achieve performance comparable to that of centralized systems. To the best of our knowledge, this is the first work that addresses the routing problem in IAB networks using deep reinforcement learning (RL). Additionally, our algorithm formulates the routing problem as a joint optimization problem, which promotes agent cooperation and reduces selfish behavior, resulting in more efficient use of network resources.

\subsection{Related Work}
Routing in networks has been and still is the subject of extensive research, as seen in  \cite{akkaya2005survey}, \cite{mammeri2019reinforcement}. There is a large body of literature on routing strategies for wireless networks, including earlier protocols such as DSR \cite{johnson1996dynamic} and AODV \cite{AODV} for ad-hoc networks, and various routing protocols for delay and disruption tolerant networks (DTNs) \cite{DTNjain2004routing}, as well as strategies for resource constrained wireless networks (such as sensor networks or internet-of-things networks) \cite{sha2013multipath}. However, many existing routing protocols were designed for specific wireless network scenarios and may not be easily adaptable to other scenarios. For example, routing protocols for ad-hoc networks assume a connected network, while routing protocols for DTNs assume a disconnected network. In this study, we focus on routing in IAB networks that are characterized by dynamic topology changes and strict time constraints, and cannot be expected to be subject to these assumptions.
This has motivated the introduction of methods that can acquire a nearly optimal policy without requiring a-priori knowledge. A major technique that is capable of achieving this goal is RL, which is a class of machine learning algorithms that can learn an optimal policy via interaction with the environment without knowledge of the system dynamics (such algorithms are also known as model-free algorithms) \cite[Ch.~1]{sutton2018reinforcement}.
One of the most popular RL techniques is Q-learning \cite{watkins1992q}, which can  learn the optimal policy online by estimating the optimal action-value function.
Early works that applied Q-learning to network routing used the classical tabular Q-learning method \cite{boyan1994packet, choi1995predictive,kumar1998confidence_q_routing}. This system enables each device to forward a limited number of packets in each time slot; for this it receives a reward over the ACK signal. However, it becomes computationally difficult to apply this method when the state space becomes large. This issue has motivated the combination of deep learning \cite{lecun2015deep} with RL, giving rise to the deep RL class of algorithms. These algorithms have attracted much attention in recent years due to their ability to approximate the action-value function for large state and action spaces. Recently, the authors in \cite{mnih2015human} proposed a  deep RL-based algorithm called deep Q-network (DQN), which combines deep neural networks and  Q-learning. Recent studies that derived DRL-based algorithms for network routing problems can be found in \cite{stampa2017deep, valadarsky2017learning, Feat_Engineering_for_DRL, you2020toward, Hierarchical_Deep_Double_Q_Routing,actor_critic_for_adaptive_routing_hybrid_method, relational_drl}: 

A DRL approach for routing was developed in \cite{stampa2017deep} with the objective of minimizing the delay in the network. In this approach, a single controller finds all the paths of all source-destination pairs given the demand in the network, which represents the bandwidth request of each source-destination pair. However, this approach results in complex state and action spaces, and does not scale for large networks as it depends on a centralized controller. Moreover, in this approach, the state representation does not capture the network topology, which is highly dynamic in our scenario. Motivated by the high complexity of the state and action representations of the approaches proposed in \cite{stampa2017deep,valadarsky2017learning}, a feature engineering approach has recently been proposed in \cite{Feat_Engineering_for_DRL} that only considers some candidate end-to-end paths for each routing request. This proposed representation was shown to outperform the representation of the approaches proposed in \cite{stampa2017deep, valadarsky2017learning} in some use-cases. In \cite{you2020toward}, the authors employed a deep recurrent Q-Network (DRQN), to determine the routing policy, this is a mixture of a DQN and a long short-term memory (LSTM) network. Using device-specific characteristics, such as the last k actions taken and the next m destinations of packets in queues, the algorithm is trained for each device. The LSTM layer in the DRQN algorithm uses past observations for the prediction of the whole network state, which, in turn, allows the agent to select the next hop for a specific packet in the next time step.
The work in \cite{Hierarchical_Deep_Double_Q_Routing} applied an algorithm called hierarchical-DQN (h-DQN) \cite{kulkarni2016hierarchical}. h-DQN facilitates exploration in complicated environments by decomposing the original problem into a hierarchy of sub-problems such that higher-level tasks invoke lower levels as if they were primitive actions. In \cite{actor_critic_for_adaptive_routing_hybrid_method}, the authors used another RL algorithm called Actor Critic algorithm, which is a policy-based RL algorithm where the policy learned directly via parametrization. They compared their results with that of the algorithm in \cite{boyan1994packet,kumar1998confidence_q_routing} and showed that their proposed algorithm achieves better performance. 

\subsection{Paper Structure and Notations}
The organization of this paper is as follows. In Section \RNum{2}, we present the problem formulation and assumptions. In Section \RNum{3}, we provide the mathematical background for our proposed solution, including Markov Decision Processes (MDPs), RL, and Multi-Agent Reinforcement Learning (MARL). Section \RNum{4} presents our mathematical formulation for the problem and explains the rationale behind the chosen MARL approach. In Section \RNum{5}, we provide a review of existing routing algorithms. Section \RNum{6} describes our proposed algorithm, which is based on the A2C method and includes three different training paradigms, ranging from fully decentralized to centralized training. Simulation results, including a comparison with existing routing algorithms, are presented in Section \RNum{7}. Finally, in Section \RNum{8}, we conclude this work and discuss future research directions.
\par
Throughout this work, we use $\mathbb{N}$ to denote natural numbers, bold letters ( e.g., $\mathbf{X}$) to denote vectors, and $\mathbf{X_i}$ denotes the $i$th element in the vector $\mathbf{X}$, $i\geq0$. Further, $\mathbf{X_{t,i}}$ denotes the $i$th element in the vector $\mathbf{X}$ at the $t$th time-step, $i\geq0, t\in\mathbb{N}$. we use calligraphic letters to denote sets, e.g., $\mathcal{X}$, and the cardinality of a set is denoted by $|\cdot|$, e.g., $|\mathcal{X}|$ is the cardinality of the set $\mathcal{X}$. Lastly, $\mathbb{E}[\cdot]$ denotes the stochastic expectation.

\section{Problem Formulation}

\begin{figure}[H]
\begin{center}
\includegraphics[width=1.00\columnwidth]{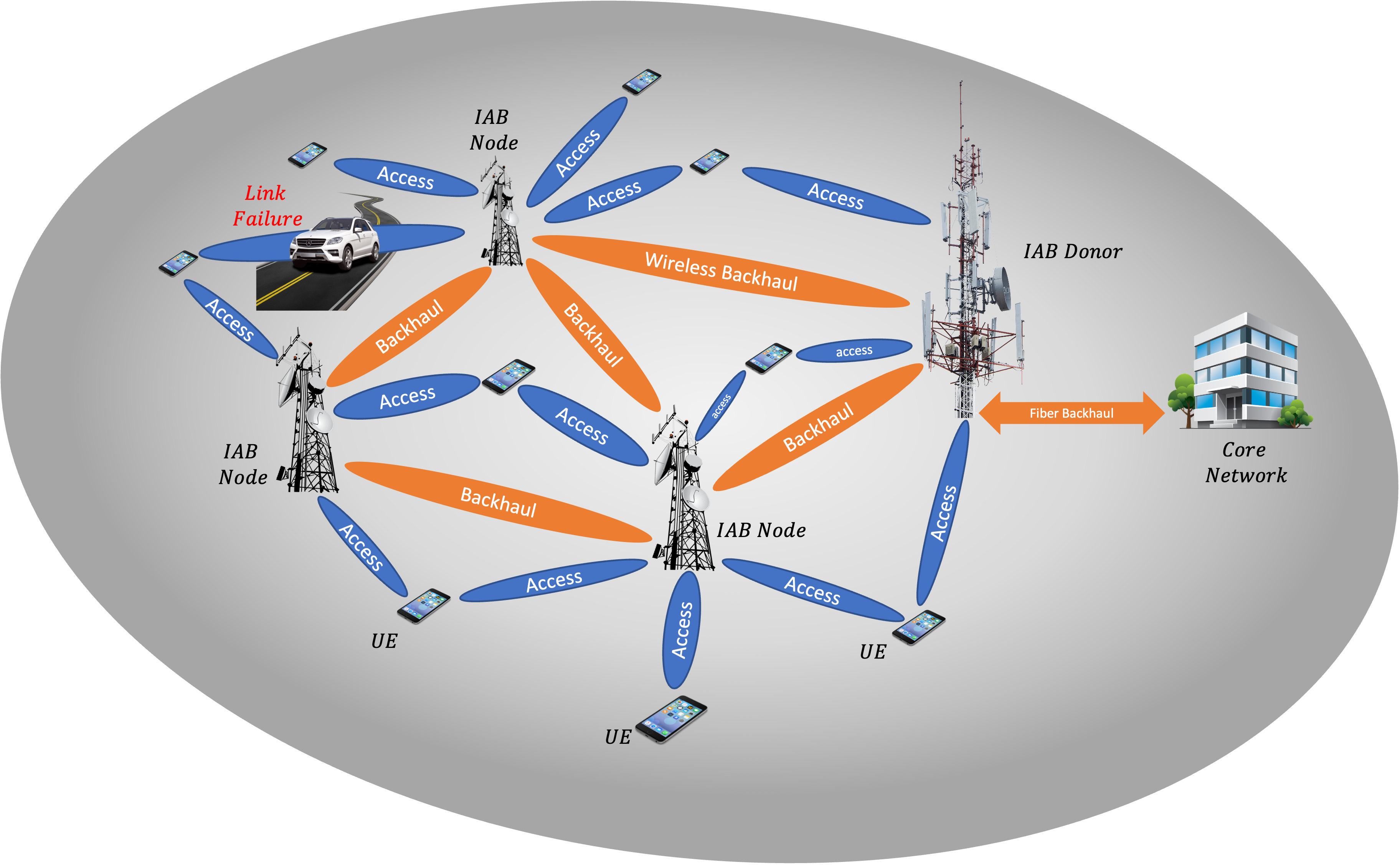}
\caption{IAB network illustration with 1 IAB donor, 3 IAB nodes and 9 UEs.}
\label{fig:problem_formulation}
\end{center}
\end{figure}
We consider a multi-hop IAB wireless network with an IAB donor, IAB nodes and User Equipments (UE) \cite{5G_IAB}, as shown in Fig. \ref{fig:problem_formulation}. The IAB donor is wired to the core network, whereas IAB nodes use wireless communication to backhaul their traffic to the core network via a multi-hop connection. Both the IAB donor and IAB nodes provide access and backhaul interfaces for the UE and IAB nodes, respectively. We model this network by an undirected weighted graph $\mathcal{G}=(\mathcal{N},\mathcal{L}, d),$ where $\mathcal{N}$, $\mathcal{L}$ denote the sets of nodes and wireless links, respectively, and $d:\mathcal{L}\rightarrow\mathbb{N}$ assigns a delay to each wireless link. There are three sets present in $\mathcal{N}$: a set $\mathcal{D}$ of the IAB donor, a set $\mathcal{B}$ of the IAB nodes and a set $\mathcal{U}$ of the UEs, i.e., $\mathcal{N} = \mathcal{D} \cup \mathcal{B} \cup \mathcal{U}$. Each of the nodes $n \in \mathcal{D} \cup \mathcal{B}$ is equipped with an independent buffer queue, and a transceiver with beam-forming and routing capabilities. Each of the links $(n,m)\in\mathcal{L}$ is a bidirectional link between node $n$ and node $m$, portraying a time-varying wireless channel. 
\par
We assume that time is slotted by $t\in\mathbb{N}$, and, for simplification, we assume that packets are constant in length and that transmission rates are limited to transmit integer numbers of packets per slot. As another simplification, we represent the wireless link's spectral efficiency as a delay between the two nodes of the graph using the mapping $d$. This assumption is based on the fact that links with varying degrees of spectral efficiency will require a different number of transmissions to transfer the same amount of data, so using a link with low spectral efficiency rather than a link with high spectral efficiency will result in an increased delay.
\par
Once an IAB node or UE is activated, it is connected to an already active node, i.e., either an IAB donor or another IAB node which has a path to an IAB donor. Thus, we build the network topology in an iterative greedy fashion similarly to in \cite{Distributed_Path_Selection}, where we set constraints over the maximal number of IAB parents ($P^{max}_{parent}$), IAB children ($C^{max}_{children}$), number of users each base station has ($U_{children}^{max}$), and the number of associated base stations each user has ($U_{parent}^{max}$). It should be noted that by using the following topology generation scheme, we receives a connected graph, i.e., there is a path from any base station to any other node in the network. We assume that all our nodes operate at mmWave bands for backhaul and access transmission and reception (in-band backhauling) with beam-forming capabilities. Therefore, similarly to in \cite{simsek2020iab}, we disregard the interference between non-assigned nodes since narrow beam mmWave frequencies have a power limit rather than an interference limit. 
\par
Considering the stochastic nature of packet arrivals, we use a Poisson process with parameter $\lambda$, which we refer to as the network load. In each time-slot we sample the arrival process, which corresponds to the number of packets generated. In order to distribute the packets across the network's base stations, the donor receives the number of packets that corresponds to its available wireless bandwidth. The remainder of the packets are distributed uniformly among the network's base stations. As each packet is generated, it is given a time limit, which is referred to as its Time To Live (TTL). If this TTL expires, the packet is dropped. To prioritize packets with lower TTLs, the base-stations use an unlimited-sized prioritized queue based on TTL. Consequently, the base station always processes packets in accordance with the prioritized queue.
\par
We denote $\mathcal{N}_i$ as the set of neighbors of node $i \in \mathcal{N}$, i.e., $(i,j) \in\mathcal{L},\text{ } \forall j\in\mathcal{N}_i$, and let $C,K\in\mathbb{N}$ be the number of wireless channels and activated base stations, respectively.
In each time step, each base station extracts a set of packets from its queue. This is followed by deciding where to send each packet, which means that the $i^{th}$ base station has to choose one destination from $\mathcal{N}_i$ for each packet. In our model, users may move or change their base-station associations between two consecutive time slots, which would be resolved in a change of the network topology.
In addition, the edge delays are slowly varying around their initial values to modulate the changes in the wireless medium.

\section{Background}
In the following section, we introduce the mathematical foundations, on which we will base our proposed solution. We begin by exploring MDPs and Partially Observe MDPs (POMDPs). Then, we describe the tools from the field of RL and MARL that we use as the basis for our method.
\subsection{MDPs and Partially Observe MDPs Preliminaries}
We can define MDP as a tuple, $<\mathcal{S},\mathcal{A},\mathcal{R},\mathcal{P}>$, where $\mathcal{S},\mathcal{A}$ and $\mathcal{R}$ are the sets of environment states, actions and rewards, respectively. In addition, let $\mathcal{P}$ be the set of the probabilities $\big{\{}\Pr(s',r|s,a)\big{\}}_{s',s\in\mathcal{S}, a\in\mathcal{A}, r\in\mathcal{R}}$, whereas the probability $\Pr(s',r|s,a)$ represents the agent probability of observing next state $s'$ and reward $r$ after being at state $s$ and performing action $a$. The agent's policy can be represented as the following mapping $\pi:\mathcal{S}\times\mathcal{A}\rightarrow [0,1]$, which represents a mapping from the current state to the probability distribution on the action space.
\par
A Partially Observe MDP (POMDP) is defined as a tuple $(\mathcal{O}, \mathcal{S}, \mathcal{A}, ,\mathcal{R}, \mathcal{P})$ \cite{puterman2014markov}.
An agent interacting with the environment at state $\textbf{s} \in \mathcal{S}$ observes an observation $o(s) \in \mathcal{O}$. After observing $o$, the agent selects an action $a \in \mathcal{A}$ based on this observation, which means that now the agent's policy is determined by its observations. That is, policy can be represented as the following mapping $\pi:\mathcal{O}\times\mathcal{A}\rightarrow [0,1]$, which represents a mapping from the current observation to the probability distribution on the action space.
After performing an action, the remainder of the decision process is the same as for the MDP case.
\par
\subsection{RL Preliminaries}
We examine a general model-free RL framework \cite{sutton2018reinforcement} applied to a specific MDP, in which the agent interacts with the environment and learns to accomplish the task at hand by a series of discrete actions. In this case, the agent does not assume any prior knowledge of the underlying environment statistics. Figure \ref{fig:decision_processes_frameworks} describes the decision making processes between the agent and the environment: at each discrete time step $t$, the agent observes the current state of the environment $S_t$ and executes an action $A_t$ according to its policy $\pi$. Then, the agent receives an immediate reward $R_t$ and the environment transitions to a new state $S_{t+1}$, based on the transition kernel $\Pr(S_{t+1}|S_{t},A_{t})$, to the next state $S_{t+1}$.

\begin{figure*}[ht!]
\begin{center}
    \label{fig:mdp_framework}{%
      \includegraphics[width=0.6\textwidth]{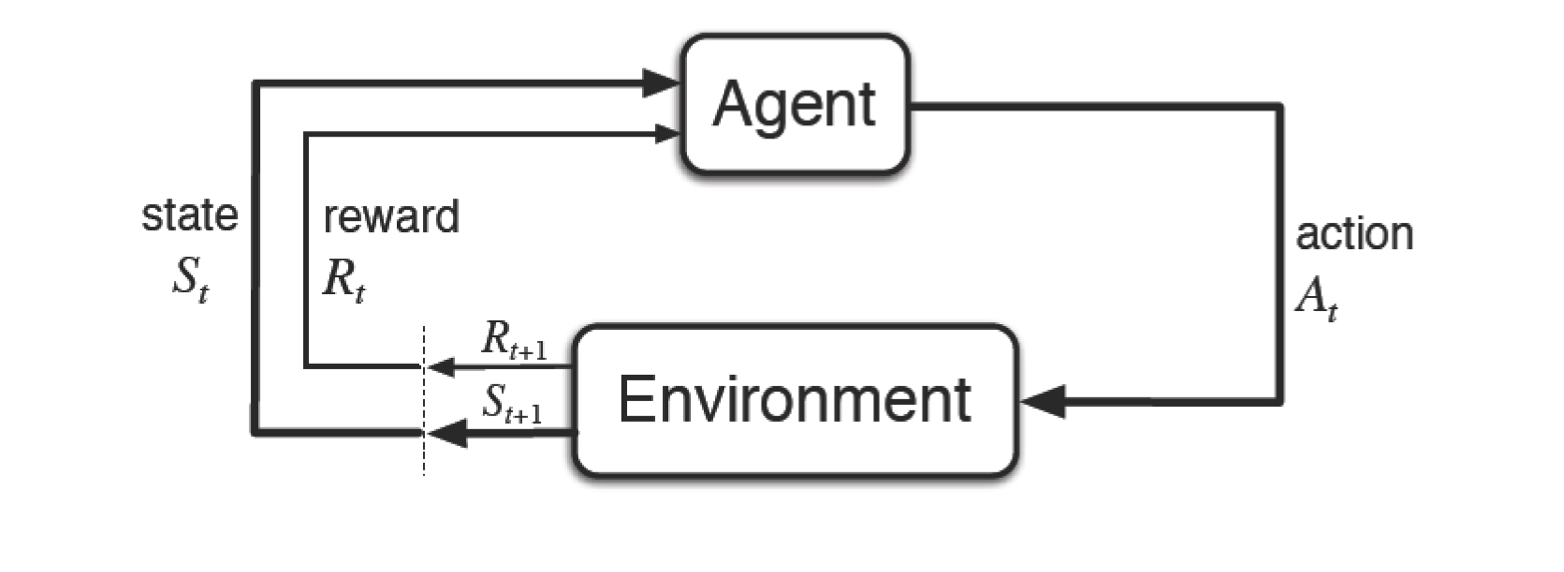}}
\caption{Decision process framework.}
\label{fig:decision_processes_frameworks}
\end{center}
\end{figure*}    

Under this framework, the agent's goal is to select actions that maximize the expected cumulative discounted reward  $G_t$, where we define $G_t$ as follows:
\begin{align}
    G_t = \sum_{n=0}^{\infty}\gamma^n R_{t+n+1},\; \gamma \in [0,1) \label{eq:1}.
\end{align}
The discount factor $\gamma$ determines how much immediate rewards are favored over more distant rewards. For a general MDP framework, the action-value function $Q_{\pi}(S_t,A_t)$ represents the expected accumulated discounted reward starting from state $S_t$, picking action $A_t$, and following policy $\pi$ afterwards, whereas the value function $V_{\pi}(S_t)$ represents the expected accumulated discounted reward starting from state $S_t$, and following policy $\pi$ \cite[Ch.~3]{sutton2018reinforcement}:
\[Q_{\pi}(S_t,A_t) \triangleq E_\pi[G_t|S_t, A_t] \text{,  } V_{\pi}(S_t) \triangleq E_\pi[G_t|S_t].\] 
The optimal policy $\pi^*$ is a policy that satisfies $Q_*(S_t,A_t)\triangleq Q_{\pi^*}(S_t,A_t) \geq Q_{\pi}(S_t,A_t)$ for any policy $\pi$ and for every possible state-action pair. By continuously interacting with the environment, the RL agent aims to learn the optimal policy $\pi^*$.
\subsection{MARL Preliminaries}
\begin{figure}
    \centering
    \includegraphics[width=0.6\textwidth]{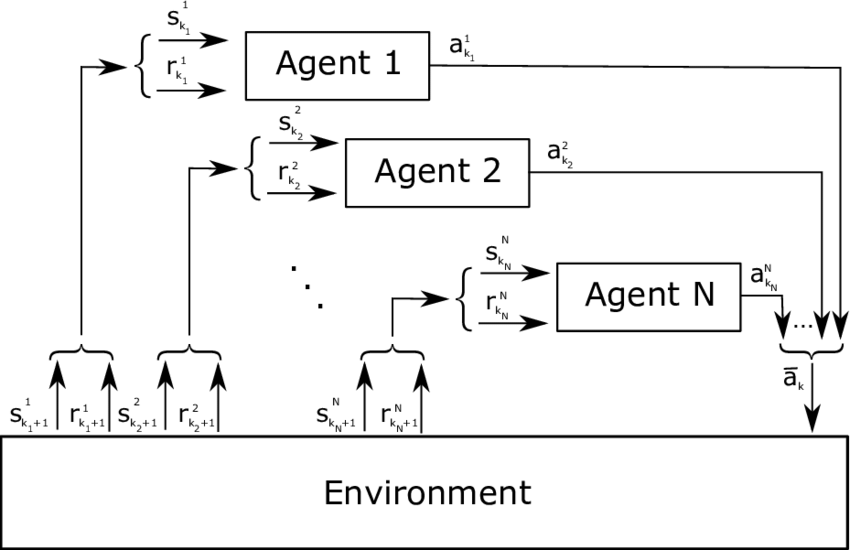}
    \caption{Multi-agent decision process framework.}
    \label{fig:multiAgentMDP}
\end{figure}
MARL also addresses sequential decision-making problems, but with more than one agent involved, as illustrated in Fig. \ref{fig:multiAgentMDP}. In particular, both the evolution of the system state and the reward received by each agent are influenced by the joint actions of all agents. More intriguingly, each agent has has to optimize its own long-term reward, which now becomes a function of the policies of all other agents \cite{marl_selective_overview}. 
\par
As we are interested in optimizing the network performance in this study, it is essential that our agents act cooperatively. In a fully cooperative setting, all agents will tend to share a common reward function, such as $R_1 = R_2 = \cdots = R_N = R$. We note that this model may also be referred to as a multi-agent MDP (MMDP) in the AI community \cite{boutilier1996planning}. In this model, the value function and Q-function are identical for all agents, which enables the application of single-agent RL algorithms, where all agents are coordinated as one decision maker \cite{marl_selective_overview}. 
\par
Besides the common-reward model, another slightly more general and emerging model for cooperative MARL considers team-average rewards \cite{zhang2018fully,marl_selective_overview}. Specifically, agents are allowed to have different reward functions, which may be kept private for each agent, while the goal for cooperation is to optimize the long-term reward corresponding to the average reward $R \triangleq \frac{1}{N}\sum_{i=0}^{N-1} R_i$. The average-reward model, which allows more heterogeneity among agents, includes the model above as a special case. It also preserves privacy among agents, and facilitates the development of decentralized MARL algorithms \cite{zhang2018fully,marl_selective_overview}.
\section{RL for IAB Network Routing}
In the following section, we formulate and evaluate the IAB network routing problem using a Multi-Agent POMDP. First, we propose and analyze an approach to solve this problem. Next, we proceed to formulate it mathematically as a Multi-Agent POMDP. We conclude our discussion by describing our evaluation metrics.
\par
During the agents' training phase, their experience is formed from tuples of $<s, a, r>$ that are chained together into time sequences by the next state, $s'$. 
We can define this experience in two ways. On the one hand, we can select the base station as our decision-maker. Hence, every base station acts as an independent agent, making its own decisions about how to forward packets. On the other hand, we can select the packet as our decision-maker. By doing this, every packet in the network is an agent and makes an independent decision when it reaches the front of a base station's queue.
In our work, we employed the packet approach to define our decision process mathematically due to its convenience. In light of the fact that our network contains multiple packets, we consider this to be a multi-agent problem.

\subsection{Formulating IAB Routing Using a Multi-Agent POMDP Framework}\label{sec:iab_routing_formulation}
This section outlines the mathematical formulation of our problem. Due to the presence of multiple agents with a joint goal, we formulate the problem of multi-agent routing as a as a Milti-Agent POMDP \cite{puterman2014markov} with discrete actions. Each agent observes statistics relevant to itself and does not observe the entire global state. We denote the the Multi-Agent POMDP by a tuple  $(\mathcal{S},\mathcal{O},\mathcal{A}_1,\mathcal{A}_2,\ldots,\mathcal{A}_N,\mathcal{P},\mathcal{R},N,\gamma)$, where $N$ is the number of agents and $\mathcal{S}$ is the environment state space. Environment state $\textbf{s} \in \mathcal{S}$ is not fully observable. Instead, agent $i$ draws a private observation $\textbf{o}_i \in \mathcal{O}$ that is correlated with $\textbf{s}$. $\mathcal{A}_i$ is the action space of agent $i$, yielding a joint action space $\mathcal{A}=\mathcal{A}_1\times \mathcal{A}_2\times\cdots\times \mathcal{A}_N$. $\Pr(\textbf{s}',r|\textbf{s},\textbf{a})\in\mathcal{P}$ is the state-reward transition probability, where $r,\textbf{s}, \textbf{s}',\textbf{a}\in\mathcal{R}\times\mathcal{S}\times\mathcal{S}\times\mathcal{A}$. $\mathcal{R},\gamma$ represent the available rewards set and the discount factor, respectively. Agent $i \in \{1,2,\cdots ,N\}$ uses a policy $\pi_i: \mathcal{O}\times \mathcal{A}_i\rightarrow [0,1]$ to choose actions after drawing observation $\textbf{o}_i$. After all agents have taken actions, the joint action $\textbf{a}$ triggers a state transition $\textbf{s} \rightarrow \textbf{s}'$ based on the state transition probability $\Pr(\textbf{s}'|\textbf{s},\textbf{a})$. 
\par
In Mulit-Agent POMDP,  MARL can be used as a computational tool. Below, we specify each component using the MARL definitions.
\begin{itemize}
    \item \textbf{Observations.}
    The agent can only observe information relevant to the packet it controls. Specifically, we consider the packet's current node, time to live (TTL), and queue delay in this work. Let $\textbf{o}_i = (n, t, QueueDelay(i,n))$ be the agent $i$ observation. The scope of $n$ in this context is limited to only the base stations and the final destination, i.e., other users cannot relay messages.
    \item \textbf{Actions.} If agent $i$ is authorized to perform a wireless hop based on $\textbf{o}_i = (n,t)$, the allowable action set includes all wireless links that are available from node $n$. For example, $a_i = l_{n,n'}$, where $l_{n,n'}\in\mathcal{L}$ represents the link between nodes $n$ and $n'$. Let $\textbf{a} = (a_1,a_2,\cdots,a_N)$ be the joint action. Whenever the agent does not have permission to conduct a wireless hop, its action will be defined as null.
    \item \textbf{Transitions}. Joint action $\textbf{a}$ triggers a state transition $\textbf{s} \rightarrow \textbf{s}'$ with probability $\Pr(\textbf{s}' |\textbf{s}, \textbf{a})$ which depends on the dynamics of the environment and on the frequency at which the agent is polled to provide an action. In this case, the frequency is dependent on the duration of each time slot in the system.
    \item \textbf{Reward.} Let $D_{i}$ represent the immediate delay of the $i^{th}$ agent. To be more specific, $D_i$ consists of two components that cause delay to the $i^{th}$ agent: the instant wireless link delay and the delay caused by waiting at the current base-station queue. We define the current agent's delay as the sum of these two terms. If we define the observation as $\textbf{o}_i=(n,t)$ and the action as $a_i=l_{n,n'}$, we can denote the immediate delay of the agent as $D_{i} = -(q^n_i + d((n,n')))$, where $q^n_i$ represents the delay induced by the $n^{th}$ node queue to the $i^{th}$ agent. Accordingly, we define the agents' joint reward using the immediate delay representation as follows: $R\triangleq\frac{1}{N}\sum_{i\in\mathcal{I}}D_i\in\mathcal{R}$, where $\mathcal{I}$ represents the set of active agents.
\end{itemize}

Let $\Pi\triangleq\times_{i=0}^{N-1}\pi_i$ represent the joint agent policy.
Let $\mathbf{s}_t\in\mathcal{S}$ denote the state of the network at time $t\in\mathbb{N}$.
Our objective is to derive an RL-based algorithm to identify the set of policies $\{\pi_i\}_{i=0}^{N-1}$ that maximize the expected accumulated discounted reward over 
a finite time horizon, ~i.e, 
\begin{equation}
\label{eqn:RL-objective}
\Pi^{\star}=\operatornamewithlimits{\rm argmax}\limits_{\Pi}\Bigg\{ \mathbb{E}_{\Pi}\bigg[\sum_{t=1}^{T}\gamma^{t-1}R_{t+1}\Big|\mathbf{s}_1\bigg]\Bigg\}.
\end{equation}

\subsection{Evaluation Metrics}
A reactive routing scheme is used in a multi-hop IAB network to dynamically minimize packet delay while ensuring that packets reach their destination on time. There might be multiple hops, links with inefficient spectral efficiency, and nodes with overloaded queues in a packet's path, all of which may cause delay. There are various metrics to measure or estimate the congestion within the network. In our scenario we define our congestion estimation using the following metrics: 
\begin{itemize}
    \item Packet latency - The time it takes for a packet to travel from its source to its destination.
    \item Arrival Ratio - The percentage of packets that made it to their destination successfully.
\end{itemize}
This multi-objective problem aims to minimize the packet latency while simultaneously maximizing the arrival ratio. Therefore, it may suffer from a Pareto-front, which means that optimizing w.r.t. one objective, leads to a sub-optimal solution w.r.t. another objective \cite{pareto_front_optimization}.
Despite the fact that the network performance measurements are well defined, an individual agent does not necessarily have access to their signals. For example, the arrival ratio represents a metric which is dependent upon the entire network; due to its multi-hop structure, at each time slot an individual cannot even obtain a good estimation of this value.
\section{Existing Routing Algorithms}
The following section describes existing solutions for solving a routing problem. Throughout this section, we explore and modify those algorithms to solve the problem of routing in an IAB network.
\subsection{Centralized-Routing} 
In this policy, the algorithm has access to the full network state. During each time step, the next hop is determined by computing the shortest path to the packet's destination. By observing the full state while calculating the shortest path, this algorithm also accounts for delays caused by queues at other base stations along the packet's path. Even though this algorithm has a high complexity, it achieves network performance improvement with full knowledge of the network's state, and serve as a benchmark in comparison to a decentralized approach. 
\subsection{Minimum-Hop Routing}
In this policy, the algorithm has access only to the links' delay \cite{bellman1958routing}. Based solely on those link delays, the next hop is determined for each time step by computing the shortest path to the packet's destination. In addition to serving as a benchmark, this algorithm is intended to analyze the influence of queue delay on the resulting delay.
\subsection{Back-Pressure Routing}
In this policy \cite{backpressure_routing}, each base-station stores a queue for each possible destination. The following procedure describes this algorithm. In each time slot, each node calculates the difference between the specific queue length and the corresponding queue length located at each of its neighboring nodes (this calculation occurs for each queue). By utilizing this difference for all available destinations, each base station is able to make two different decisions. First, it makes a greedy decision regarding which packets to extract from the queues based on the size of the differentiation ( in other words, the destinations with the greatest degree of differentiation will be selected). Second, the node will determine the packet's next hop based on a greedy decision (the neighbor with the highest differentiation score).

\subsection{Q-Routing} 
In this policy, each node uses an off policy iterative method termed Q-learning \cite{boyan1994packet}. Next, Q-learning is first used to learn a representation of the network state in terms of Q-values and then these values are used to make routing decisions. Given this representation of the state, the action $a$ at node $n$ is to find the best neighbor node to deliver the a packet that results in lower latency for the packet to reach its destination.
The following procedure describes this algorithm.
As soon as node $n$ extracts a packet destined to node $d$ from his queue, it selects the next-hop decision based on the $\epsilon$-greedy policy w.r.t. its Q-values, i.e.,
\begin{equation}
    y=
    \begin{cases}
        UniformRandom(\mathcal{A}_{n}) & \text{\textit{w.p.} } \epsilon, \\       \operatornamewithlimits{argmax}\limits_{a\in\mathcal{A}_{n}}\hat{Q}^{(n)}(a,d) & \text{\textit{w.p.} } 1-\epsilon.
    \end{cases}
\end{equation}
Once the packet has been sent from node $n$ to neighboring node $y$, node $y$ will send its best estimate back to node $n$, $\operatornamewithlimits{max}\limits_{a'\in\mathcal{A}_{y}} \hat{Q}^{(y)}(a',d)$ for the destination $d$ back to node $n$ over the ACK signal. This value essentially estimates the remaining time in the journey of the packet. Following that, the Q-value is modified by the following formula:
\begin{equation}
\hat{Q}^{(n)}_{new}(y,d) = \hat{Q}^{(n)}_{old}(y,d) + \alpha\cdot\big(r + \gamma\cdot \operatornamewithlimits{max}\limits_{a'\in\mathcal{A}_{y}} \hat{Q}^{(y)}(a',d) - \hat{Q}^{(n)}_{old}(y,d)\big).
\end{equation}

\subsection{Full Echo Q-Routing}
This algorithm is a variant of the Q-Routing algorithm described above, and addresses the well-known issue of exploration against exploitation \cite{boyan1994packet}; similarly to Q-Routing, it utilizes an iterative off-policy method termed Q-Learning. In order to eliminate the problem of exploration against exploitation, this algorithm executes the following procedure. As soon as node $n$ extracts a packet destined to node $d$ from its queue, it selects the next-hop decision based on the greedy policy w.r.t. its Q-values, $y=\operatornamewithlimits{argmax}\limits_{a\in\mathcal{A}_{n}}\hat{Q}^{(n)}(a,d)$. Once the packet has been sent from node $n$ to neighboring node $y$, node $n$ receives information from all neighboring nodes, which transmit their respective best estimates for packets destined to node $d$, i.e., 
$\operatornamewithlimits{max}\limits_{a'\in\mathcal{A}_{y}} \hat{Q}^{(y)}(a',d), \forall y\in\mathcal{N}_n$. Using those estimations, node $n$ modifies its Q-values for each neighbor by utilizing the following formula:
\begin{equation}
\hat{Q}^{(n)}_{new}(y,d) = \hat{Q}^{(n)}_{old}(y,d) + \alpha\cdot\big(r + \gamma\cdot\operatornamewithlimits{max}\limits_{a'\in\mathcal{A}_{y}}{\hat{Q}}^{(y)}(a',d) - \hat{Q}_{old}^{(n)}(y,d)\big), \forall y\in\mathcal{N}_n.
\end{equation}
\subsection{Hybrid Routing}
In this algorithm, a Q-Routing agent is trained simultaneously with an on policy iterative method \cite{actor_critic_for_adaptive_routing_hybrid_method}. In this case, Q-learning is used to learn a representation of the network state in terms of Q-values, and then Hybrid routing uses these values to update the agent's policy parameters by using the Actor-Critic method. As soon as node $n$ extracts a packet destined to node $d$ from its queue, it selects its next action based on sampling its policy distribution, $\pi_n(d,\cdot;\boldsymbol{\theta}_n)$. Then, it sends the packet to one of its neighboring nodes $y$.
The corresponding Q-value is updated based on a Q-Routing update rule, and then its policy parameters $\boldsymbol{\theta_n}$ are updated according to the following formula:
\[
\boldsymbol{\theta_n} \leftarrow \boldsymbol{\theta_n} + \alpha \cdot \nabla_{\boldsymbol{\theta}_n}\log\pi_n(y,d;\boldsymbol{\theta}_n)\cdot\big(r+\gamma\max_{a'\in\mathcal{A}_y} \hat{Q}^{(y)}(a',d)-\max_{a\in\mathcal{A}_n}\hat{Q}^{(n)}(a,d)\big).
\]

\section{Our Proposed Multi-Agent Relational A2C Routing Algorithm}
In this study, we are primarily interested in exploring a decentralized solution to the problem of routing in an IAB network. In the following section, we present our novel solutions to this issue. First, we discuss the motivation behind our proposed solution and the challenges it seeks to address. Following this, we describe the main characteristics of our solution. Finally, we present three different training paradigms within our approach, ranging from fully decentralized training to centralized training. 
\par
As a first step, we attempted to solve the above task using traditional RL techniques such as Q-routing\cite{boyan1994packet}, Full-Echo Q-Routing\cite{boyan1994packet}, and Hybrid Routing\cite{actor_critic_for_adaptive_routing_hybrid_method}. These methods have not achieved a high degree of generalization due to the challenges posed by the above task, such as partial observability, a large state space, and multi-agent optimization. Essentially, these methods assume that each agent acts independently and does not share their experience with other agents, resulting in performance degradation as a result of insufficient correlation between reward signals, network state, and other agents' policies.
\par
Our proposed solution addresses these issues by formulating this problem as a Mutli-Agent POMDP as described in Sec. \ref{sec:iab_routing_formulation}. We define our algorithm objective in the same manner as described in Eq. (\ref{eqn:RL-objective}) to encourage cooperation among the different agents. Furthermore, we leverage the homogeneity between destinations in order to support an invariant number of users, as we might encounter in a real-world scenario. To this end we categorize destinations into groups based on their \textit{relational} base-station association, i.e.,
which base station they are currently connected to. Through this categorization, policy and value functions are shared by each group. As part of this process, each agent uses an iterative online on-policy method called Advantage Actor-Critic (A2C)\cite[Ch.~13]{sutton2018reinforcement}. According to this scheme, the actor decides what action to take and the critic informs the actor of its effectiveness and how it should be adjusted. Based on the current observation, critic produces an estimated representation of the network state, and the actor uses this information to update its policy. Every agent in the network represents its own strategy through its actor. To contend with the issue of the large size of the state space, we propose using neural networks to approximate both the actor and the critic. Consequently, since we categorize agents according to their relational association with base stations, we refer to this algorithm as \textit{Relational A2C}.

In this section, we present three algorithms based on our method, each in its own subsection. The first algorithm is Relational A2C, in which training is centralized. Specifically, all agents use the same global actor and global critic. Figure \ref{fig:Illustration_relational} shows how these are updated based on information from all of the IAB stations. The second algorithm is the Dec-Relational A2C algorithm, in which each base station has a local actor and critic that are trained based on the local base station experience. Figure \ref{fig:Illustration_dec_relational} illustrates how the information from the base station is used to update these models.
Lastly, the third algorithm is based on federated learning \cite{fedAVGpaper}, which we refer to as Fed-Relational A2C. Most of the time, this algorithm operates in a decentralized manner, but once within a given period of time, it converts the weighted averages of these local model weights into a global model, which then is broadcast to the base stations.

\begin{figure*}[h!]
  \subfloat[Relational A2C Learning Diagram. \label{fig:Illustration_relational}]{%
      \includegraphics[ width=0.4\textwidth]{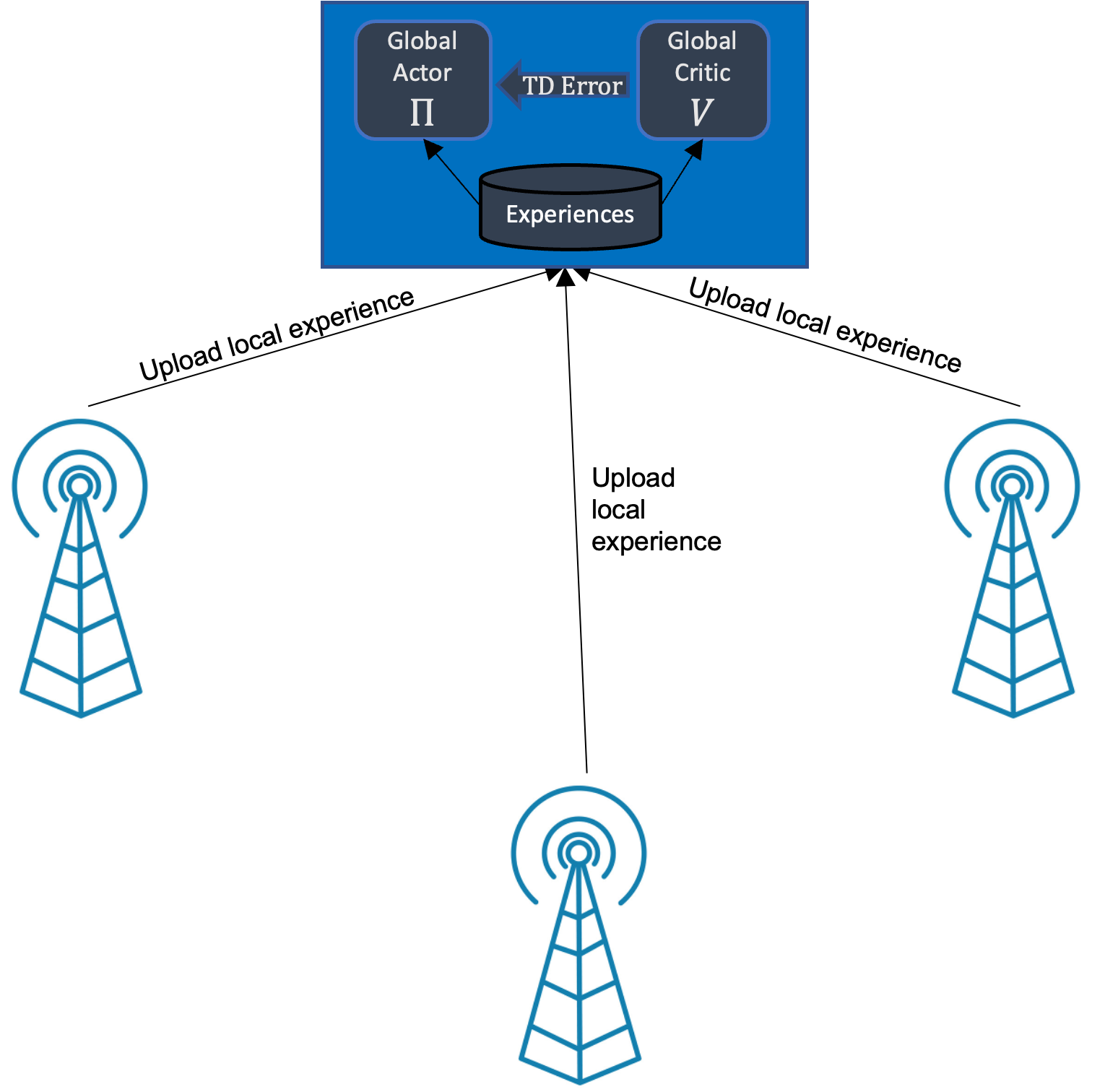}}
\hspace{\fill}
  \subfloat[Dec-Relational A2C Learning Diagram. \label{fig:Illustration_dec_relational}]{%
      \includegraphics[ width=0.49\textwidth]{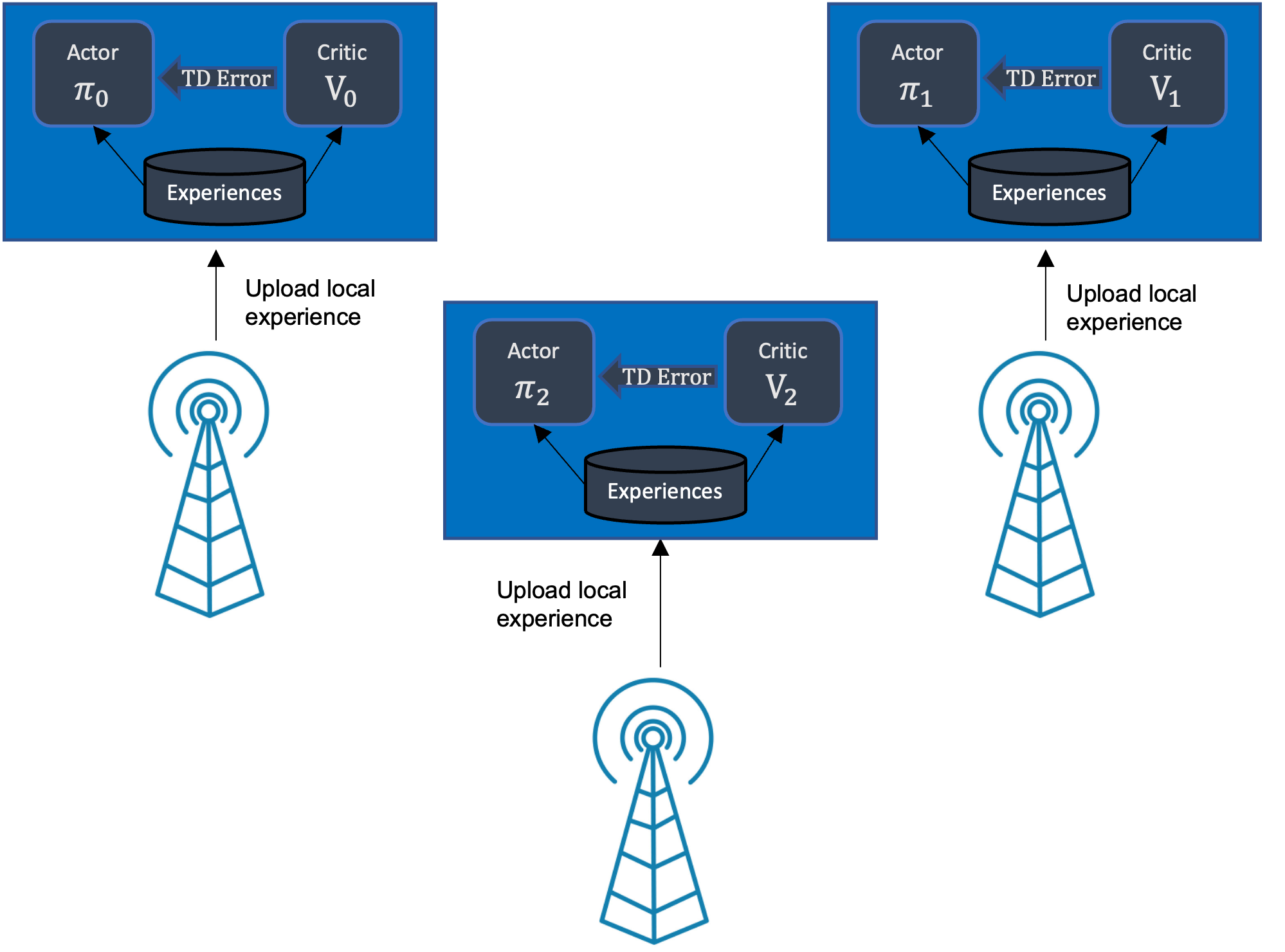}}\\
\caption{Illustration of centralized and decentralized training paradigms for our proposed method.}
\label{fig:Illustration_dec_relational_and_relational}
\end{figure*}

\subsection{Relational A2C}
In this section, we present the general idea behind our method and explain our centralized training solution. As a first step, we unify all packets destined to the same destination in accordance with the same policy. In the following step, for any possible graph topology, denoted as $\mathcal{L}^\star\triangleq\{\mathcal{L}|\mathcal{L}\text{ }is\text{ }a\text{ }set\text{ }of\text{ }graph\text{ }edges\text{ }that\text{ }forms\text{ }network\text{ }topology\}$, we define a mapping $H:\mathcal{N}\times\mathcal{L}^\star\rightarrow\{0,1\}^{K}$ that maps agents destined to nodes $n\in\mathcal{N}$ to a given group.
In this case, $K$ and $|\{0,1\}^{K}|$ represent the numbers of available base stations and different groups that our agents are divided into, respectively. Next, we define our categorization mapping, 
$H:\mathcal{N}\times\mathcal{L}^\star\rightarrow\{0,1\}^{K}$ as follows:
\begin{equation*}
    H(i,\mathcal{L})[j] \triangleq 
    \begin{cases}
        1, & \text{if } (i,j) \in \mathcal{L},\\
        0, & \text{if } (i,j) \notin \mathcal{L}.
    \end{cases}
\end{equation*}
\begin{figure}[!t]
\begin{center}
\includegraphics[width=0.75\textwidth]{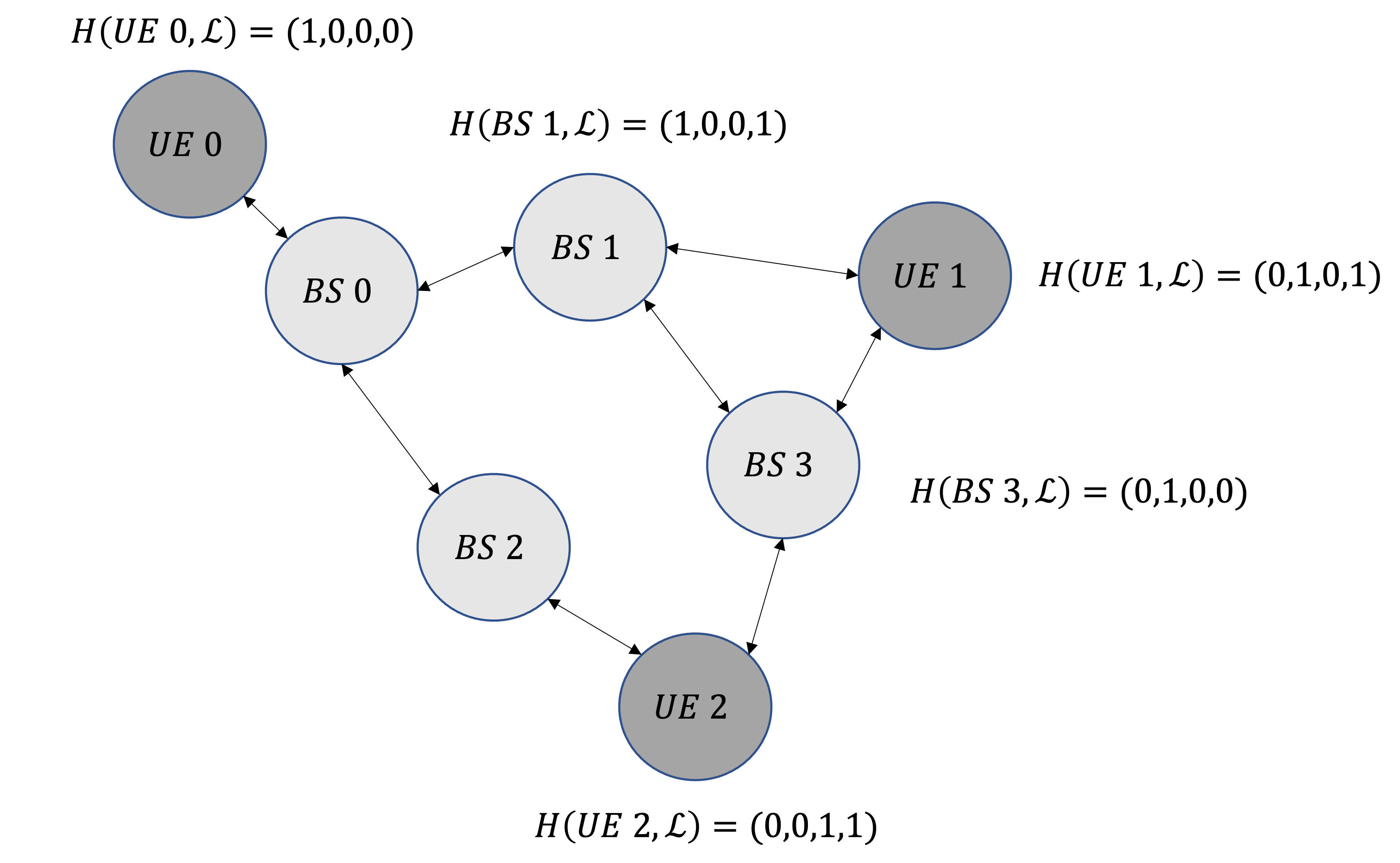}
\caption{An illustration of the categorization mapping $H$ for a scenario with $K=4$.}
\label{fig:categorization_mapping}
\end{center}
\end{figure}
\par
It is noteworthy that according to this mapping destinations are grouped according to their relational base-station associations. Furthermore, Fig. \ref{fig:categorization_mapping} demonstrates the functionality of this mapping. In addition, the division of our system's agents into different groups is intended to result in a significant reduction in the number of agents. As a result, multiple agents can cooperate under the same goal and share information with similar groups, thereby suppressing the non-stationarity issue associated with multi-agent systems \cite{actor_critic_for_load_balancing}. 
\par
Next, we present the centralized training paradigm, namely \textit{Relational A2C} (Fig. \ref{fig:Illustration_relational}). In view of the fact that we are utilizing neural networks, we would like to prevent bias in our input and categorize agents in a meaningful way. Thus, we propose encoding the observations of the $n^{th}$ agent under the assumption that it is located at the $i^{th}$ node as follows:
\begin{equation}
    \textbf{o}_n \triangleq [oneHot(i), t, H(destination(n),\mathcal{L}), QueueDelay(n,i)],
\end{equation}
where we define $oneHot:\mathcal{N}\rightarrow \{0,1\}^{K}$ as follows:
\begin{equation*}
    oneHot(n)[j] \triangleq 
    \begin{cases}
        1, & j = n,  \\
        0, & \text{otherwise}.
    \end{cases}
\end{equation*}
The last variables, $t \text{ and } QueueDelay(n,i)$, represent the TTL and the queue delay of agent $n$ at the node $i$, respectively.
\par
Let $\Pi(\cdot;\boldsymbol{\theta}), \hat{V}_{\Pi}(\cdot;\textbf{w}), \boldsymbol{\theta} \in \mathbb{R}^{L_1}, \textbf{w} \in \mathbb{R}^{L_2}, L_1,L_2 \ll |\mathcal{A}|\cdot|\mathcal{S}|$ be our actor and critic representations, respectively, as depicted in Fig \ref{fig:actor_critic_network_architecture}. In spite of the fact that we are only provided with one actor and one critic for all groups, we are still able to differentiate between them due to the mapping $H$. 
Considering this representation, and assuming the packet $n$ is located at the $i^{th}$ base station, the packet strategy determines the next hop decision by sampling the corresponding actor distribution $\Pi(\textbf{o}_n;\boldsymbol{\theta})$. Next, the $i^{th}$ base station sends the agent via a packet to its next hop decision. After that, the base station receives feedback through the ACK signal that contains the critic's estimation for the next hop and the agent's instant delay, $\hat{V}_{\Pi}(\textbf{o}'_n;\textbf{w})$ and $D_n$, respectively. The critic's value represents the next-node estimate of remaining time in a packet's journey to the agent destination, when the agent follows policy $\Pi$. 
\begin{figure}[!t]
\begin{center}
\includegraphics[width=\textwidth]{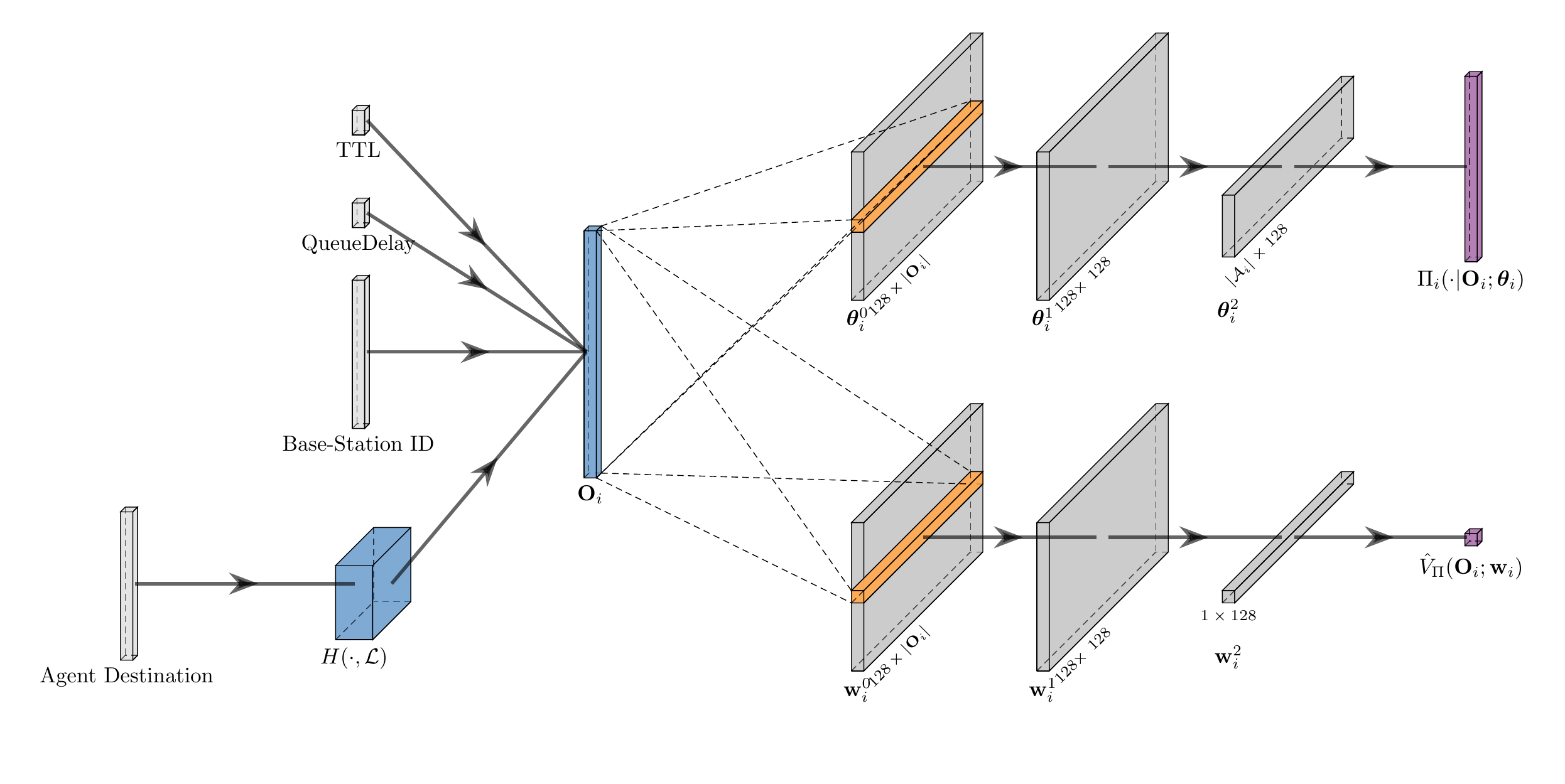}
\caption{Relational A2C neural architecture located at the $i^{th}$ base station.}
\label{fig:actor_critic_network_architecture}
\end{center}
\end{figure}

\par
Next, we consider $\hat{V}_{\Pi}(\textbf{o}_{t,n};\textbf{w})$ as the estimation of the $n^{th}$ agent averaged delay given the current observation, i.e., $\mathbb{E}[\sum_{k=0}^{T-t-1}\gamma^{k}\cdot D_{n,t+1+k}|\textbf{o}_{t,n},\Pi]$. We derive the following equality based on the Bellman equation for a single agent scenario applied for the $n^{th}$ agent:
\[
\hat{V}_{\Pi}(\textbf{o}_{t,n};\textbf{w})= \mathbb{E}[D_{n,t+1} + \gamma\cdot\hat{V}_\Pi(\textbf{o}_{t+1,n};\textbf{w})|\textbf{o}_{t,n},\Pi].
\]
Hence, we propose to update the critic by minimizing the estimation error.
This is achieved by minimizing the following objective $\mathcal{T}$ w.r.t. $\textbf{w}$:
\begin{equation} \label{eq:critic_gradient}
\mathcal{T}(\Pi)=\mathbb{E}\bigg[\bigg(\hat{V}_{\Pi}(\textbf{O}_{n};\mathbf{w}) - (D_{n} +\gamma\cdot \hat{V}_{\Pi}(\textbf{O}_{n}';\mathbf{w})\big)\bigg)^2\bigg|\Pi\bigg],
\end{equation}
where $\textbf{O}_{n}, D_{n} \text{ and } \textbf{O}'_{n}$ represent the random variables of observation, instant delay and next observation for a random agent. We update the critic's parameters using the gradient descent method:
\[
\textbf{w} \leftarrow \textbf{w} -\alpha\cdot\nabla_{\mathbf{w}}\mathcal{T}(\Pi),
\]
for some learning rate $\alpha\in(0,1)$. Next, we seek to identify a set of policy parameters $\boldsymbol{\theta}$ that maximizes the expected cumulative discounted reward (refer to the definition in (\ref{eqn:RL-objective})). 
To do so, we derive the following lemma.
\begin{lemma}\label{lemma:gradient_relational}
    The gradient of the objective
    $J(\Pi)$ w.r.t. $\boldsymbol{\theta}$ is proportional to the following estimator: \begin{align*}\label{eq:actor_gradient} \nabla_{\boldsymbol{\theta}}J(\Pi)\propto
    \mathbb{E}\bigg[\sum_{n'\in\mathcal{I}_t}\nabla_{\boldsymbol{\theta}}\log\Pi(a_{t,n'}|\textbf{o}_{t,n'};\boldsymbol{\theta})\big(\big( Q_\Pi(\textbf{s}_{t},\textbf{a}_t) - \frac{1}{|\mathcal{I}_t|}\sum_{n\in\mathcal{I}_t}\hat{V}_\Pi(\textbf{o}_{t,n};\mathbf{w})\big)\bigg],
    \end{align*}
    where $\mathcal{I}_t$ represents the set of agents that were active at timestep $t$.
\end{lemma}
\begin{proof}
The full proof is attached in Appendix \ref{appendix:lemma1}.
\end{proof}

According to its definition, the action-value function $Q_\Pi(\textbf{s}_t,\textbf{a}_t)$ represents the average delay of all packets while the network is at state $\textbf{s}_t$, uses action $\textbf{a}_t$, and follows the policy $\Pi$. Our proposal is to estimate the average network delay $\hat{Q}_\Pi(\textbf{s}_{t},\textbf{a}_t)$ using our active agents in order to accommodate the partially observed issue, i.e., 
\[\hat{Q}_\Pi(\textbf{s}_{t},\textbf{a}_t)=\frac{1}{|\mathcal{I}_t|}\sum_{n\in\mathcal{I}_t}\hat{Q}_\Pi(\textbf{o}_{t,n},a_{t,n})=\frac{1}{|\mathcal{I}_t|}\sum_{n\in\mathcal{I}_t}D_{n,t+1}+\gamma\cdot \hat{V}_\Pi(\textbf{o}_{t+1,n};\textbf{w})   
.\]
Following that, we can estimate the gradient of the objective $J(\Pi)$ w.r.t. $\boldsymbol{\theta}$ at the $t^{th}$ time-step using the following sample: \begin{equation}\label{eq:actor_gradient} \nabla_{\boldsymbol{\theta}}\hat{J}(\Pi)=\sum_{i\in\mathcal{I}_t}\bigg(\frac{1}{|\mathcal{I}_t|}\sum_{n\in\mathcal{I}_t}D_{n,t+1}+\gamma\cdot \hat{V}_\Pi(\textbf{o}_{t+1,n};\textbf{w}) - \hat{V}_{\Pi}\big(\textbf{o}_{t,n};\mathbf{w}\big)\bigg)\cdot\nabla_{\boldsymbol{\theta}}\text{log}\Pi(a_{t,i}|\textbf{o}_{t,i};\boldsymbol{\theta}).
\end{equation}
This estimation allows us to apply the stochastic gradient ascent method in order to find the local optimal solution. That is, at each time step t, the parameter $\boldsymbol{\theta}$ is updated by
\[
\boldsymbol{\theta} \leftarrow\boldsymbol{\theta} + \eta \nabla_{\boldsymbol{\theta}}\hat{J}(\Pi),
\]
for some learning rate $\eta\in(0,1)$.
By using this method, we expect to increase the cooperation between the different agents through the optimization of the joint objective (Fig. \ref{fig:Illustration_relational} illustrates this procedure). The steps of the proposed Relational A2C algorithm are summarized in Algorithm \ref{alg:rel_a2c} below:
\begin{algorithm}[H]
\caption{The Relational A2C Algorithm for Simultaneously Optimize Routing Strategy}\label{alg:rel_a2c}
\begin{algorithmic}[1]

\State Initialize Actor and Critic weights $\boldsymbol{\theta},\mathbf{w}$ and learning rates $\eta,\alpha$.

\For{time step $t = 1,2, \dots , T$}
    \State Extract active agents $\mathcal{I}_{t}$.
    \For{Agent $n \in \mathcal{I}_t$}
    \State Observe $\textbf{o}_{t,n}$.
    \State $\mathbf{a}_{n,t} = 
     a, \text{w.p. } \Pi(a|\textbf{o}_{t,n};\boldsymbol{\theta}), \forall a\in\mathcal{A}_{n,t},$
    \EndFor
\State Execute actions.
\For{Agent $n \in \mathcal{I}_t$}
\State Obtain the delay $D_{n,t+1}$ and next observation $\textbf{o}_{t+1,n}$ associated with the $n^{th}$ agent.
            \State Set Temporal Difference Error.
            \State
            $\delta_{t,n} = D_{n,t+1} + \gamma \cdot \hat{V}(\textbf{o}_{t+1,n};\mathbf{w}) - \hat{V}\big(\textbf{o}_{t,n};\mathbf{w}\big)$.
\EndFor
\State Update critic and actor parameters based on Eq. \ref{eq:critic_gradient} and Eq. \ref{eq:actor_gradient}.
\State $\boldsymbol{\theta}_{t+1} \leftarrow \boldsymbol{\theta}_t + \eta \cdot \bigg(\sum_{n\in\mathcal{I}_t}\nabla_{\boldsymbol{\theta}_t} \text{log}\big(\Pi(a_{t,n}|\textbf{o}_{t,n};\boldsymbol{\theta})\big)\cdot\big(\frac{1}{|\mathcal{I}_t|}\sum_{n'\in\mathcal{I}_t}\delta_{t,n'}\big)\bigg).$
\State  $\boldsymbol{w}_{t+1} \leftarrow \boldsymbol{w}_t - \alpha \cdot\nabla_{\boldsymbol{w}}\bigg(\sum_{n\in\mathcal{I}_t}(\delta_{t,n})^2\bigg).$
\EndFor
\end{algorithmic}
\end{algorithm}

\subsection{Dec-Relational A2C}
In this section we extend our method to support a decentralized training paradigm, namely, \textit{Dec-Relational A2C}. For decentralized training, we propose decoupling the neural network across the different base stations so that each base station uses its own set of weights for both actor and critic, as depicted in Fig \ref{fig:actor_critic_network_architecture}. With a dedicated neural network for each base station, the base station ID in the observation becomes redundant, so we drop it, leading to a new observation $\textbf{o}_n = [t,H(destination(n),\mathcal{L}),QueueDelay(n))]$, where $QueueDelay$ represents the amount of time the $n^{th}$ agent waited at the current base station queue. Let $\hat{\Pi}_k(\cdot;\boldsymbol{\theta}_k), \hat{V}_{\Pi}(\cdot;\textbf{w}_k), \boldsymbol{\theta}_k \in \mathbb{R}^{L_1}, \textbf{w}_k \in \mathbb{R}^{L_2},$ be our actor and critic representations at the $k^{th}$ base-station, respectively. 
In addition, let $k,n,k',\textbf{o}_n,\textbf{o}_n'$ be our current node, agent, next node decision, observation and next step observation, respectively. Following that, based on the Bellman equation we propose to modify the temporal difference estimation as follows:
\begin{equation}
    \delta_{k,n,k'} = D_n + \gamma \cdot \hat{V}_\Pi(\textbf{o}_n';\textbf{w}_{k'}) - \hat{V}_\Pi(\textbf{o}_n;\textbf{w}_{k}),
\end{equation}
such that there is a mutual update between the different base stations through the ACK signal. Essentially, by using this modification we hope that the next base-station estimation combined with the instant delay will represent the estimated path delay from the current base station for this agent.
Next, let $\mathcal{I}_k$ represent the set of agents who took a next-hop decision from base station $k$. The critic is then updated at each base station $k$ by minimizing the following objective $\mathcal{T}_k$ w.r.t. $\mathbf{w}_k$:
\begin{equation} \label{eq:dec_critic_gradient}
    \mathcal{T}_k(\hat{\Pi}_k) =\mathbb{E}\bigg[\sum_{n\in\mathcal{I}_k}\delta_{k,n,k'_n}^2\bigg|\hat{\Pi}_k\bigg],
\end{equation}
where we update critic's parameters using the gradient descent method:
\[
\textbf{w}_k \leftarrow \textbf{w}_k -\alpha\cdot\nabla_{\mathbf{w}_k}\mathcal{T}_k(\hat{\Pi}_k),
\]
for some learning rate $\alpha\in(0,1)$. As a next step, we aim to approximate the objective gradient w.r.t. $\boldsymbol{\theta}_k$.
\begin{lemma} \label{lemma:dec_gradient_approx}
    The gradient of the objective
    $J(\Pi)$ w.r.t. $\boldsymbol{\theta}_k$ is proportional to the following estimator: \begin{align*}\label{eq:actor_gradient} \nabla_{\boldsymbol{\theta}_k}J(\Pi)\propto
    \mathbb{E}\bigg[\sum_{n\in\mathcal{I}_{k,t}}\nabla_{\boldsymbol{\theta}_k}\log\Pi(a_{t,n}|\textbf{o}_{t,n};\boldsymbol{\theta}_k)\big( Q_\Pi(\textbf{s}_t,\textbf{a}_t) - \frac{1}{|\mathcal{I}_{k,t}|}\sum_{n'\in\mathcal{I}_{k,t}}\hat{V}_\Pi(\textbf{o}_{t,n};\mathbf{w}_k)\big)\bigg],
    \end{align*}
    where $\mathcal{I}_{k,t}$ represents the set of agents that were active at the $k^{th}$ base station at timestep $t$.
\end{lemma}
\begin{proof}
    The full proof is attached in Appendix \ref{appendix:lemma2}.
\end{proof}

To support decentralized training, we first proposed a local estimation of average network delay, $\hat{Q}^{(k)}_\Pi(\textbf{s}_{t},\textbf{a}_t) =\frac{1}{|\mathcal{I}_{k,t}|}\sum_{n\in\mathcal{I}_{k,t}} D_{n,t+1} +\gamma\cdot\hat{V}_\Pi(\textbf{o}_{t+1,n};\textbf{w}_{a_{t,n}})$. 
Using this estimation and following Lemma \ref{lemma:dec_gradient_approx},  we can estimate the objective $J(\Pi)$ gradient w.r.t. $\boldsymbol{\theta}_k$ using the following sample:
\begin{equation} \label{eq:dec_actor_gradient}
\nabla_{\boldsymbol{\theta}_k}\hat{J}(\hat{\Pi}_k)=
\sum_{n\in\mathcal{I}_{k}}\nabla_{\boldsymbol{\theta}_n} \text{log}\big(\hat{\Pi}_k(a_{n}|\textbf{o}_{n};\boldsymbol{\theta}_k)\big)\big(\frac{1}{|\mathcal{I}_k|}\sum_{n'\in\mathcal{I}_{k}}\delta_{k,n',a_{n'}}\big),
\end{equation}

where we neglect the time indexing for notational simplicity. Afterwards, we update our agents policies using the gradient ascent method: 
\[
\boldsymbol{\theta}_k \leftarrow\boldsymbol{\theta}_k + \eta \nabla_{\boldsymbol{\theta}_k}\hat{J}(\hat{\Pi}_k),
\]
for some learning rate $\eta\in(0,1)$. We term this method as \textit{Dec-Relational A2C}. By following this method we are able achieve fully decentralized training of our network (Fig. \ref{fig:Illustration_dec_relational} illustrates this procedure). 
The steps of this algorithm are summarized in Algorithm \ref{alg:dec_rel_a2c} below:
\begin{algorithm}[H]
\caption{The Dec-Relational A2C Algorithm for Simultaneously Optimize Routing Strategy}\label{alg:dec_rel_a2c}
\begin{algorithmic}[1]
\State Initialize Learning Rates $\eta,\alpha$.
\For{Base Station $k = 0,1,\dots,K-1$}
\State Initialize Actor and Critic weights $\boldsymbol{\theta}_k,\mathbf{w}_k$.
\EndFor
\For{time step $t = 1,2,...,T$}
    \For{Base Station $k = 0,1,\dots,K-1$}
        \State Extract active agents $\mathcal{I}_{k,t}$.
        \For{Agent $n \in \mathcal{I}_{k,t}$}
        \State Observe $\textbf{o}_{t,n}$.
        \State $a_{t,n} = 
         a, \text{w.p. } \hat{\Pi}_k(a|\textbf{o}_{t,n};\boldsymbol{\theta}_k), \forall a\in\mathcal{A}_{n,t},$
        \EndFor
    \EndFor
\State Execute actions.
\For{Base Station $k = 0,1,\dots,K-1$}
\For{Agent $n \in \mathcal{I}_{k,t}$}
\State Obtain the delay $D_{n,t+1}$ and next observation $\textbf{o}_{t+1,n}$ associated with the $n^{th}$ agent.
            \State Set Temporal Difference Error.
            \State
            $\delta_{t,n} = D_{n,t+1} + \gamma \cdot \hat{V}(\textbf{o}_{t+1,n};\mathbf{w}_{a_{t,n}}) - \hat{V}\big(\textbf{o}_{t,n};\mathbf{w}_k\big)$.
\EndFor
\EndFor
\State Update critics' and actors' parameters based on Eq. \ref{eq:dec_critic_gradient} and Eq. \ref{eq:dec_actor_gradient}.
\For{Base Station $k = 0,1,\dots,K-1$}
\State $\boldsymbol{\theta}_{t+1,k} \leftarrow \boldsymbol{\theta}_{t,k} + \eta \cdot \bigg(\sum_{n\in\mathcal{I}_{k,t}}\nabla_{\boldsymbol{\theta}_k} \text{log}\big(\hat{\Pi}_k(a_{t,n}|\textbf{o}_{t,n};\boldsymbol{\theta}_{t,k})\big)\cdot\big(\frac{1}{|\mathcal{I}_{k,t}|}\sum_{n'\in\mathcal{I}_{k,t}}\delta_{t,n'}\big)\bigg).$
\State  $\boldsymbol{w}_{t+1,k} \leftarrow \boldsymbol{w}_{t,k} - \alpha \cdot\nabla_{\boldsymbol{w}_k}\bigg(\frac{1}{|\mathcal{I}_{k,t}|}\sum_{n\in\mathcal{I}_{k,t}}(\delta_{t,n})^2\bigg).$
\EndFor
\EndFor
\end{algorithmic}
\end{algorithm}

\subsection{Fed-Relational A2C}
To conclude, we propose another version that combines the features of the previous versions. As part of this approach, the weights of the network are updated using a federated learning approach \cite{fedAVGpaper}. Therefore, we refer to this method as \textit{Fed-Relational A2C}.
In this approach the agents constantly update their weights in a decentralized manner, similar to the decentralized approach, while the agents periodically report their weights to a central controller; based on the weights they submit, the controller calculates a new set of shared weights. The final step is to relay the updated weights to the base stations.
Thus, we find this method to be somewhere in the middle between the previous centralized and decentralized approaches. Additionally, in this approach, weights are shared periodically among the different base stations, resulting in the same observations as in the relational approach.
The steps of this algorithm are summarized in Algorithm \ref{alg:fed_rel_a2c} below:
\begin{algorithm}[H]
\caption{The Federated-Relational Advantage Actor Critic (Fed-Relational A2C) Algorithm for Simultaneously Optimize Routing Strategy}\label{alg:fed_rel_a2c}
\begin{algorithmic}[1]
\State Initialize Actor and Critic weights and learning rates $\boldsymbol{\theta},\mathbf{w}, \eta,\alpha$.
\State Initialize federated update period $\tau$.

\For{BS $k = 0,1,\dots,K-1$} \Comment{ Broadcast Actor and Critic weights.}
\State $\mathbf{w}_k \leftarrow \mathbf{w}$,     $\boldsymbol{\theta}_k \leftarrow \boldsymbol{\theta}$.
\EndFor
\For{time step $t = 1,2, \dots, T$}
\If{$t \text{ MOD } \tau == 0$} 
        \State Collect Actor and Critic weights from each base-station $\{\mathbf{w}_k\}_{k=0}^{K-1}, \{\boldsymbol{\theta}_k\}_{k=0}^{K-1}$ with the number of updates since last global update $\{n_k\}_{k=0}^{K-1}$.
        \State \Call{FederatedUpdate}{$\{\mathbf{w}_k\},\{ \boldsymbol{\theta}_k\}_{k=0}^{K-1}$,K,$\{n_k\}_{k=0}^{K-1}$}.
\EndIf 
\State Follow decision and training phases of Decentralized Algorithm \ref{alg:dec_rel_a2c}.\Comment{Lines 6-25.}
\EndFor
\Procedure{FederatedUpdate}{CriticWeights,ActorWeights,K,NumberOfUpdates}
\State $\{\textbf{w}_k\}_{k=0}^{K-1} \leftarrow \text{CriticWeights}$, \Comment{Unpack the Critic Weights}.
\State $\{\boldsymbol{\theta}_k\}_{k=0}^{K-1} \leftarrow \text{ActorWeights}$, \Comment{Unpack the Actor Weights}.
\State $\{n_k\}_{k=0}^{K-1} \leftarrow \text{NumberOfUpdates}$ \Comment{Unpack each BS's update number.}
\State $\mathbf{w}\leftarrow\sum_{k=0}^{K-1}\frac{n_k}{\sum_{k=0}^{K-1}n_k}\mathbf{w}_k$,  $\boldsymbol{\theta}\leftarrow\sum_{k=0}^{K-1}\frac{n_k}{\sum_{k=0}^{K-1}n_k}\boldsymbol{\theta}_k,$ \Comment{Apply FedAvg Update Rule.}
\For{BS $k = 0,1,\dots,K-1$} \Comment{Broadcast the new calculated weights to the agents}
    \State $\mathbf{w}_k \leftarrow \mathbf{w}$,  $\boldsymbol{\theta}_k \leftarrow \boldsymbol{\theta}$.
\EndFor
\EndProcedure
\end{algorithmic}
\end{algorithm}
\par
It should be mentioned that in a general POMDP setting solved using MARL techniques, as considered here, a common problem is multiple local optimum points \cite{hu2003nash} within the joint policy space, which may be resolved with convergence to a less desirable, local optima strategy solution \cite{hu2003nash}. As a result, convergence to the optimal policy is not guaranteed theoretically. In practice, however, it achieves very good performance even in various POMDP models with infinitely large state space. For example, the work in \cite{foerster2018counterfactual} developed an Actor Critic algorithm for teaching multiple agents how to play Starcraft games directly from screen images, and achieved very good performance at various stages.

\section{Experiments and Insights}
In this section, we describe our main research insights and their associated experiments. First, we evaluate several connectivity scenarios in order to demonstrate the importance of network routing in an IAB network. Next, we study the impact of network load as well as the impact of mobility. This analysis shows that high loads affect network performance significantly in the case of mobility for most of the routing algorithms. We then analyze how changes in online traffic patterns and node failure affect the network. As a result of this experiment, we gained valuable insights into how the proposed routing can adapt to and recover from online changes. Lastly, we analyzed the different routing convergence times 
\par
As part of these experiments, we study and evaluate various routing methods within an IAB network, as discussed in Section \RNum{5}. In particular, the performance of Relational A2C is compared to six other algorithms: Centralized Routing, Minimum-Hop Routing, Back-pressure Routing, Q-Routing, Full Echo Q-Routing, and Hybrid Routing. We refer the reader to Section \RNum{5} for a detailed explanation of each algorithm.
\par
To conduct those experiments, we have developed a gym-based simulated IAB environment \cite{brockman2016openai}. The simulation takes place over a 2-dimensional grid \footnote{For reproducing our results, we refer the reader to run our code, https://github.com/Shahaf-Yamin/Routing-In-IAB-Networks}. Tables \ref{tab:network_hyperparameters},\ref{tab:algorithm_hyperparameters} describe network and algorithms hyper parameters, respectively.
Furthermore, Relational A2C was implemented as described in Algorithms \ref{alg:rel_a2c},\ref{alg:dec_rel_a2c},\ref{alg:fed_rel_a2c} in Section \RNum{6}, with all methods based on the A2C network architecture shown in Fig \ref{fig:actor_critic_network_architecture}. In the following, all the metrics we mentioned in Section \RNum{4} are used as the figure-of-merit for evaluating the performance of the different algorithms.

\begin{table}[!h]
\begin{center}
 \begin{tabular}{| c | c | c | c | c | c | c | c | c | c | c |}
 \hline
 $|\mathcal{D}|$ & $|\mathcal{B}|$ & $|\mathcal{U}|$ & $N$ & $P^{max}_{parent}$ & $C^{max}_{children}$ & $U_{children}^{max}$ & $U_{parent}^{max}$ & $TTL$ & $\text{User's Speed } [\frac{m}{sec}]$ & $C$\\  
 \hline
  1 &  9 & 100 & 300000 & 3 & 3 & 35 & 2 & 50 & 3 & 1\\ 
 \hline
\end{tabular}
\caption{Simulation network hyper-parameters.}
\label{tab:network_hyperparameters}
\end{center}
\end{table}

\begin{table}[!h]
\begin{center}
 \begin{tabular}{| c | c | c | c | c | c |}
 \hline
 $\gamma$ & $\alpha$ & $\eta$ & $\epsilon_D$ & $\epsilon_{min}$ & $\text{Federated Update Frequency [Time-slots]}$ \\
 \hline
  0.995 &  0.0001 & 0.0001 & 0.9999 & 0.01 & 1000\\ 
 \hline
\end{tabular}
\caption{Algorithm's hyper-parameters.}
\label{tab:algorithm_hyperparameters}
\end{center}
\end{table}
\begin{figure}[H] 
  \centering
  \includegraphics[width=0.6\textwidth]{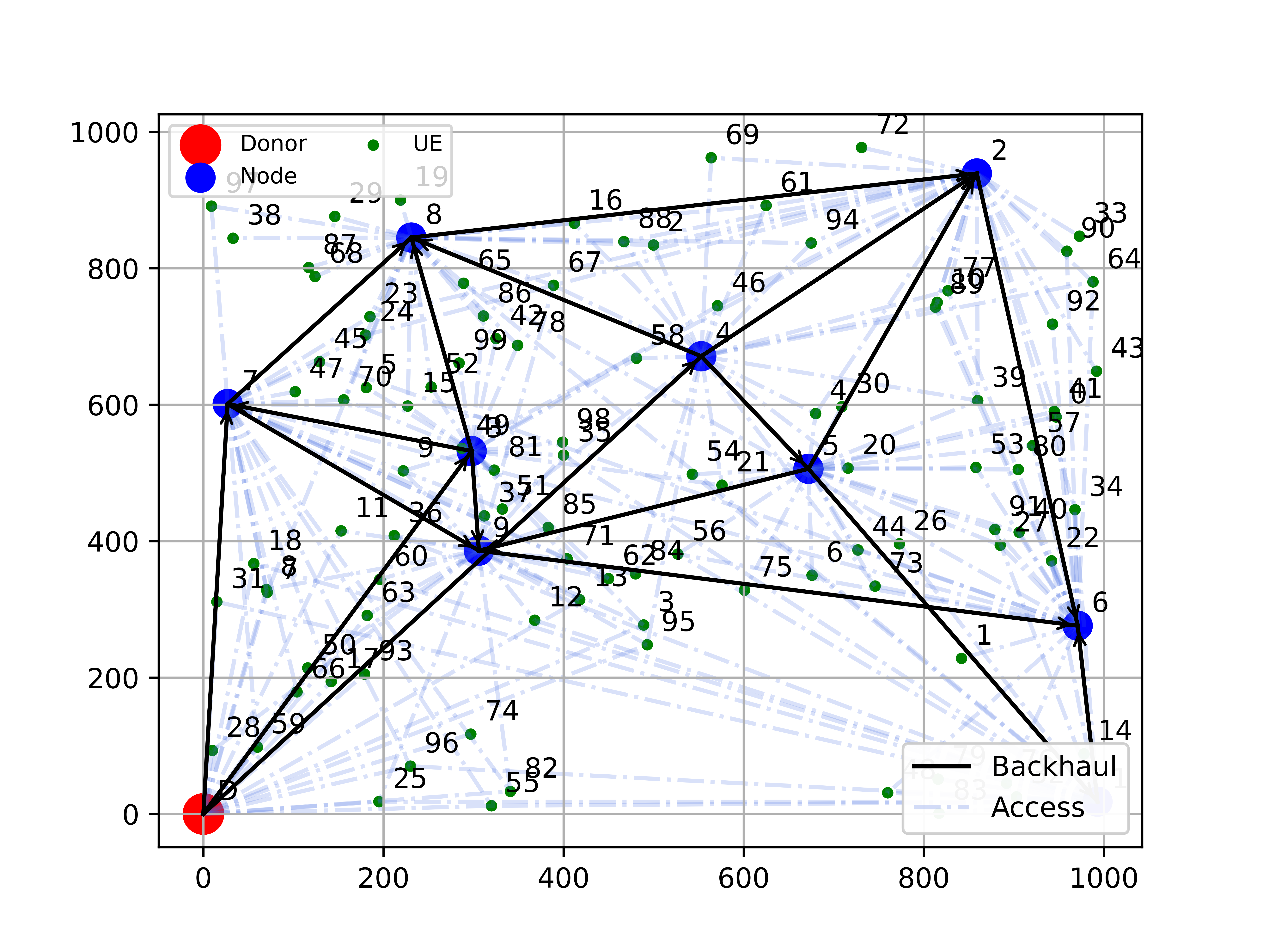}
  \caption{Network Topology Illustration}
  \label{fig:Illustration_topology}
\end{figure}

Figure \ref{fig:Illustration_topology}. illustrates a network composed with $1$ IAB Donor, $9$ IAB Nodes and $100$ users. 

\subsection{The Importance of Routing in an IAB network}
In this section, we examine the importance of routing in an IAB network by measuring network performance in various scenarios of connectivity. Our study emphasizes the importance of routing algorithms by illustrating that for high connectivity, where routing is needed, performance is much higher than for low connectivity, where routing is not needed due to the limited number of paths. In order to change the network's connectivity we have modified the constraints that dictate the number of parents each IAB node / User may have and the number of devices (IAB children / Users) that each IAB can support.

\begin{figure*}[ht!]
  \subfloat[Single parent topology (no routing). \label{fig:single_parent_topology}]{
      \includegraphics[width=0.31\textwidth]{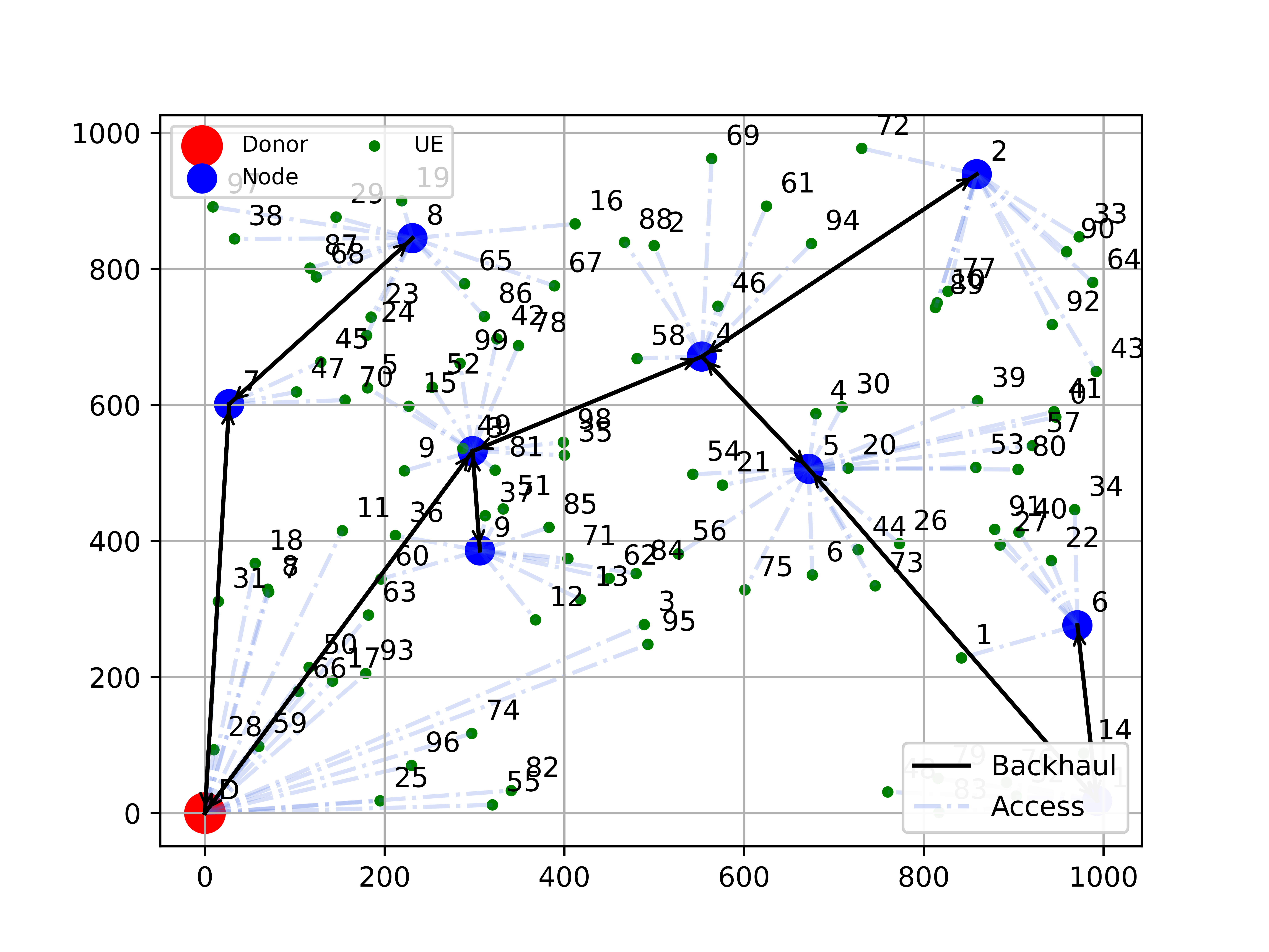}}
\hspace{\fill}
  \subfloat[3 Parents topology \label{fig:PKT} ]{
      \includegraphics[width=0.31\textwidth]{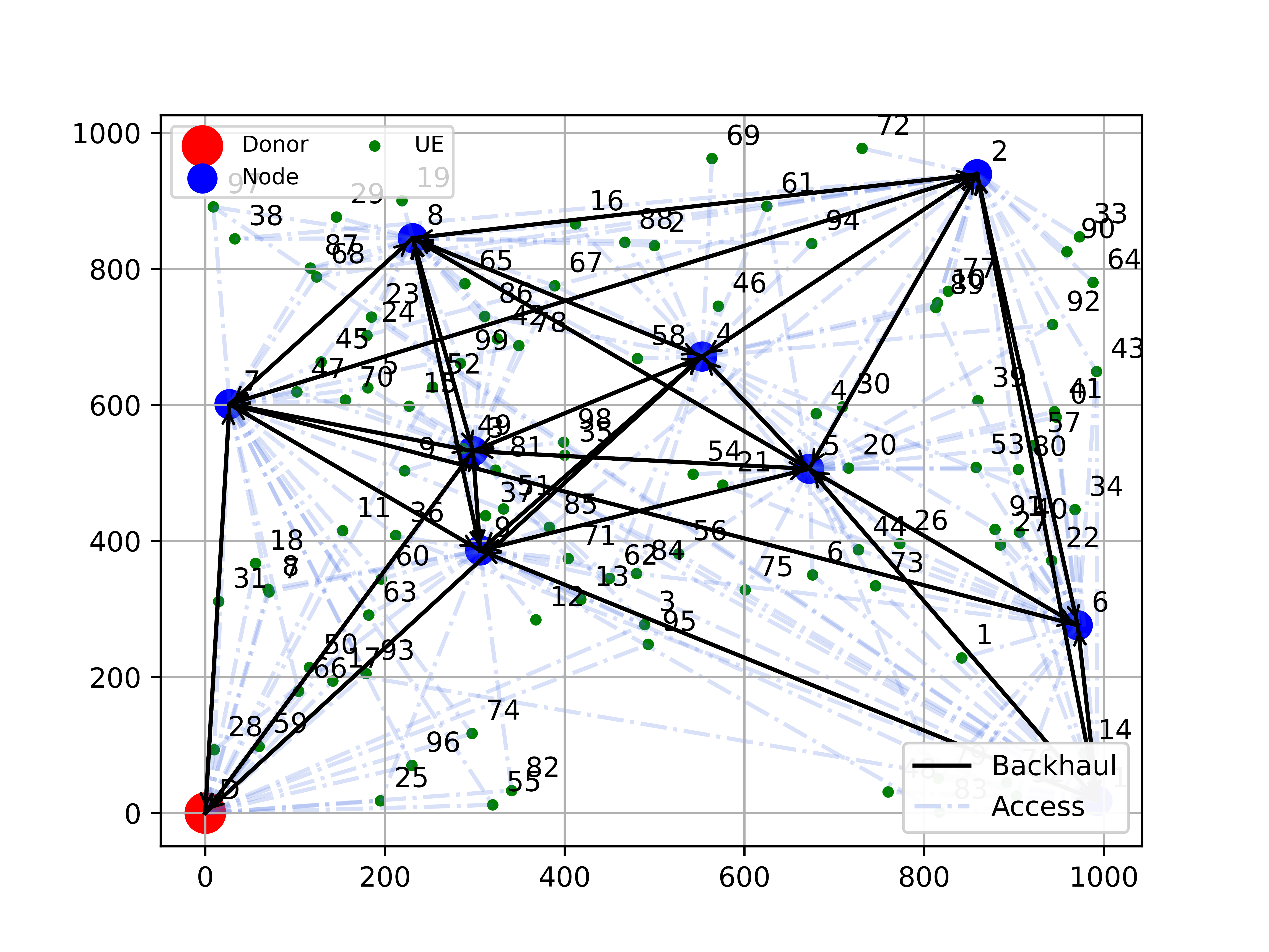}}
\hspace{\fill}
  \subfloat[Full-connectivity topology (mesh). \label{fig:tie5}]{
      \includegraphics[width=0.31\textwidth]{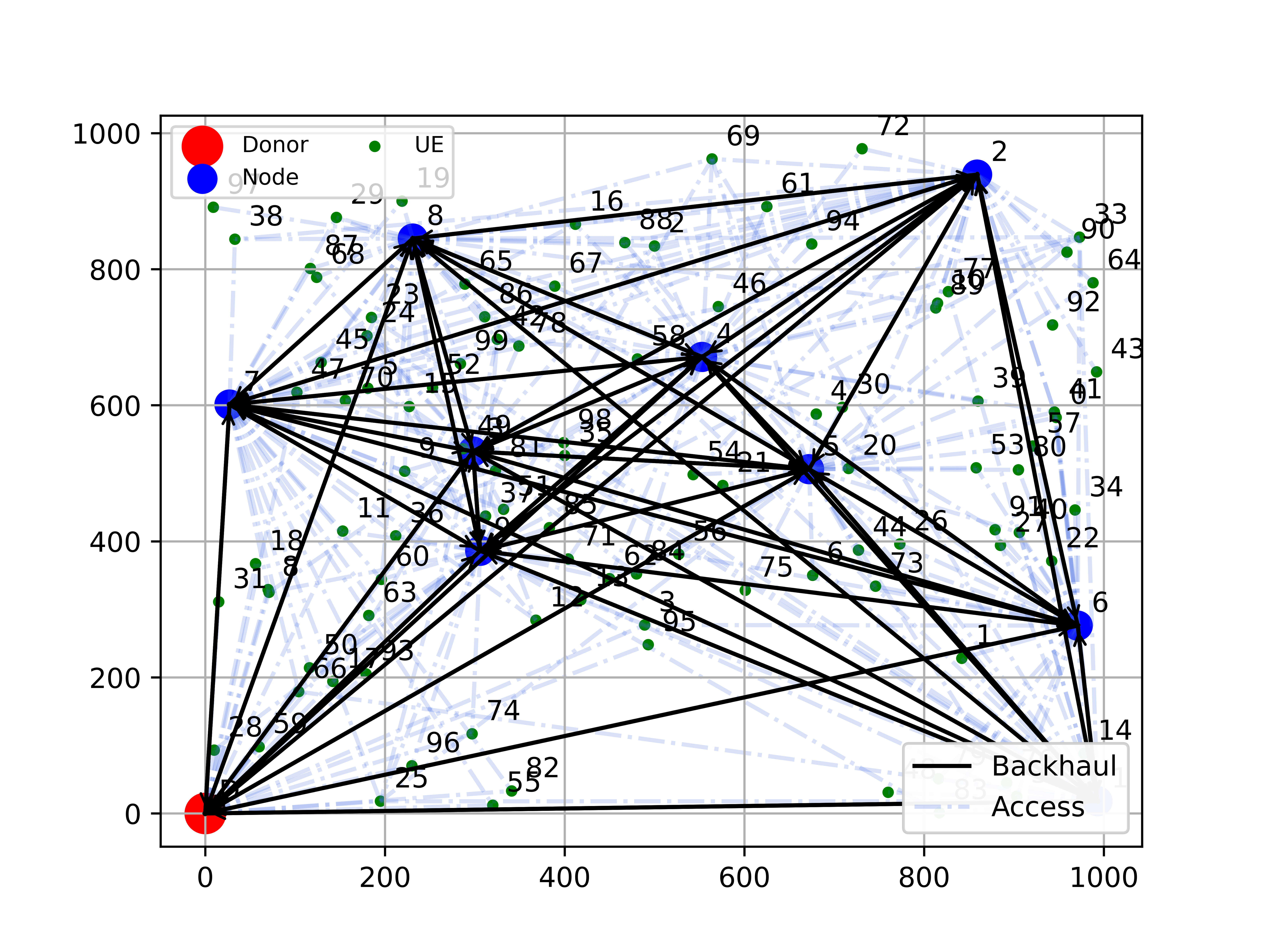}}\\
\caption{Visualization of different connectivity scenarios.}
    \label{fig:connectivity_visualization}
\end{figure*}

\begin{figure*}[ht!]
\hspace{\fill}
  \subfloat[Arrival ratio performance. \label{fig:PKT} ]{%
      \includegraphics[ width=0.49\textwidth]{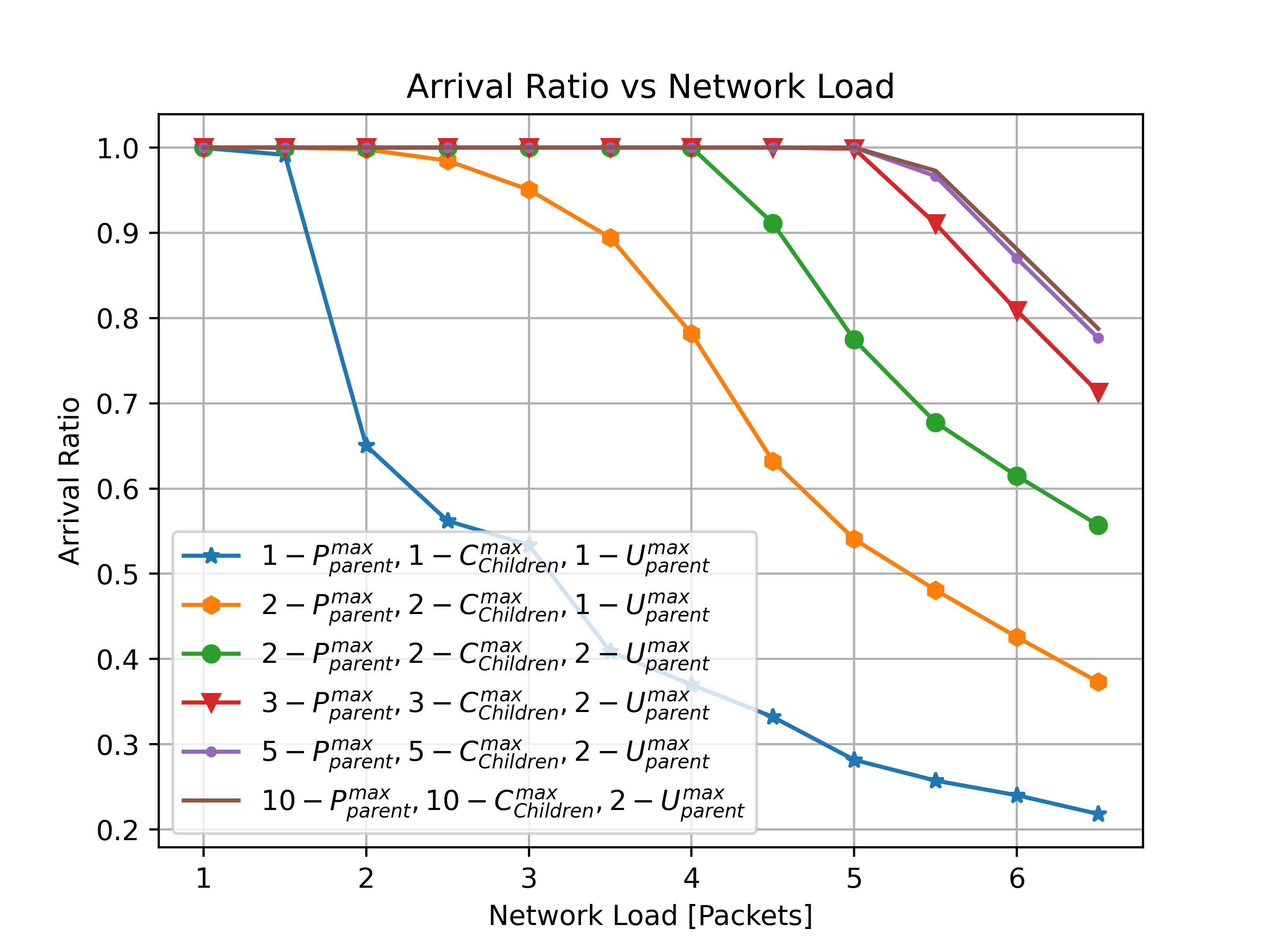}}
\hspace{\fill}
  \subfloat[Average delay time performance. \label{fig:tie5}]{%
      \includegraphics[ width=0.49\textwidth]{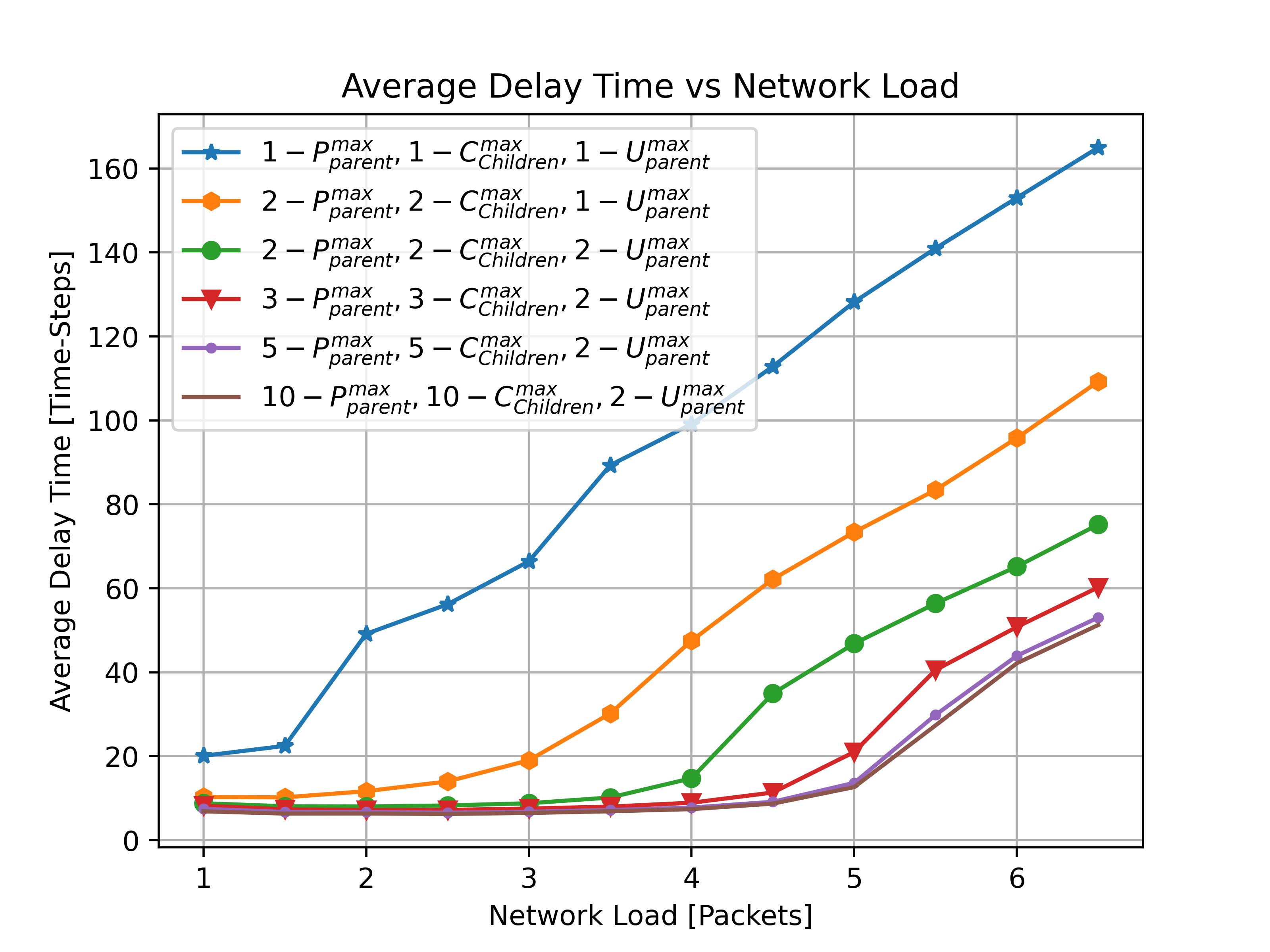}}\\
\caption{Performance illustration of different network's connectivity cases.}
    \label{fig:connectivity_influence}
\end{figure*}

Using the results illustrated in Fig. \ref{fig:connectivity_influence}, we can conclude that using a single path topology (Fig \ref{fig:connectivity_visualization}.a.) yielded the lowest arrival ratio and the longest delay. We can also deduce that higher network connectivity is highly recommended to support the expected 5G's requirements of high load and low latency. Moreover, while higher connectivity may have some benefits, we must also consider that a higher level of connectivity may result in more interference between nearby base stations, which may contradict our assumption regarding interference. In light of this conclusion, the purpose of this study is to examine how using routing can improve network performance, while only taking partial connectivity into account.

\subsection{The Influence of Network Load}
This section examines how network load affects the performance of different algorithms in various mobility scenarios. That is, we evaluate each algorithm's resilience under different loads for various mobility scenarios. Following this, we conclude that an increase in load / mobility affects network performance significantly for most of the proposed routing algorithms. As a next step, we will describe in detail our experiments and their corresponding empirical results.
\par
To change the network load, the parameter $\lambda$ of the Poisson distribution has been modified. This parameter indicates the average number of packets generated in each time-slot by the IABs. To modify $\lambda$, we scanned various loads successively from bottom-to-top, and then from top-to-bottom. The results presented are an average of 10 different measurements for each load across five different network topologies.
\par
First, we evaluate the performance of our algorithms over a static topology, which means that users cannot change their base-station association or location. Next, we increase the user's speed to $3_{\frac{m}{sec}}$, which means that users can change their location and change their association with a base station at any time. We then refer to this scenario as a dynamic topology and evaluate our algorithms' performance. In addition, we investigate how the speed of UEs affects the performance of the different algorithms. Furthermore, due to the lack of further information provided by the arrival ratio, we present only the average delay metric. The following sections provide results and insights from those experiments.

\begin{figure*}[ht!]
  \subfloat[Static Topology. \label{fig:load_analysis_average_delay_static_topology}]{%
      \includegraphics[ width=0.49\textwidth]{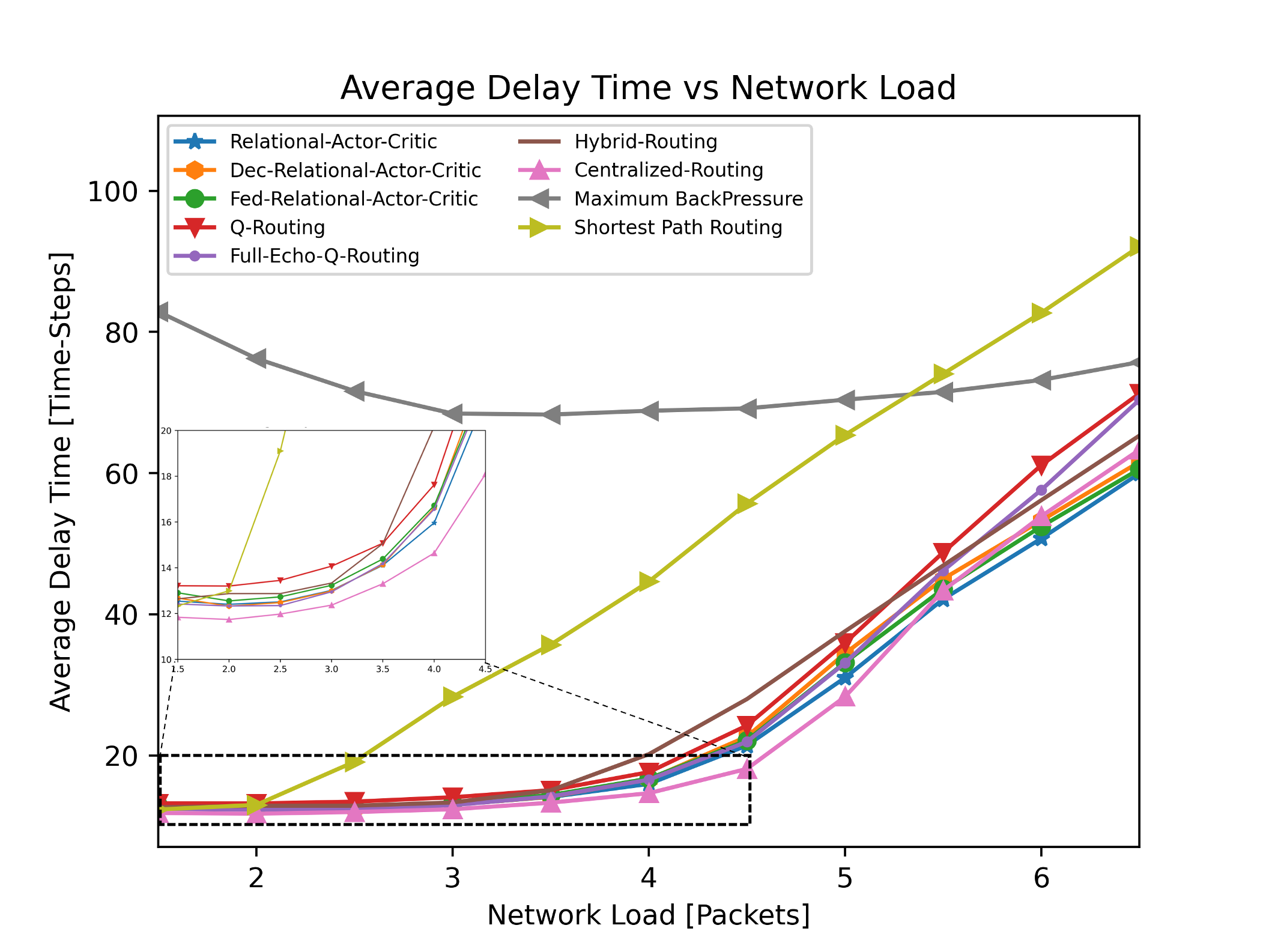}}
\hspace{\fill}
  \subfloat[Dynamic Topology. \label{fig:load_analysis_average_delay_dynamic_topology}]{%
      \includegraphics[ width=0.49\textwidth]{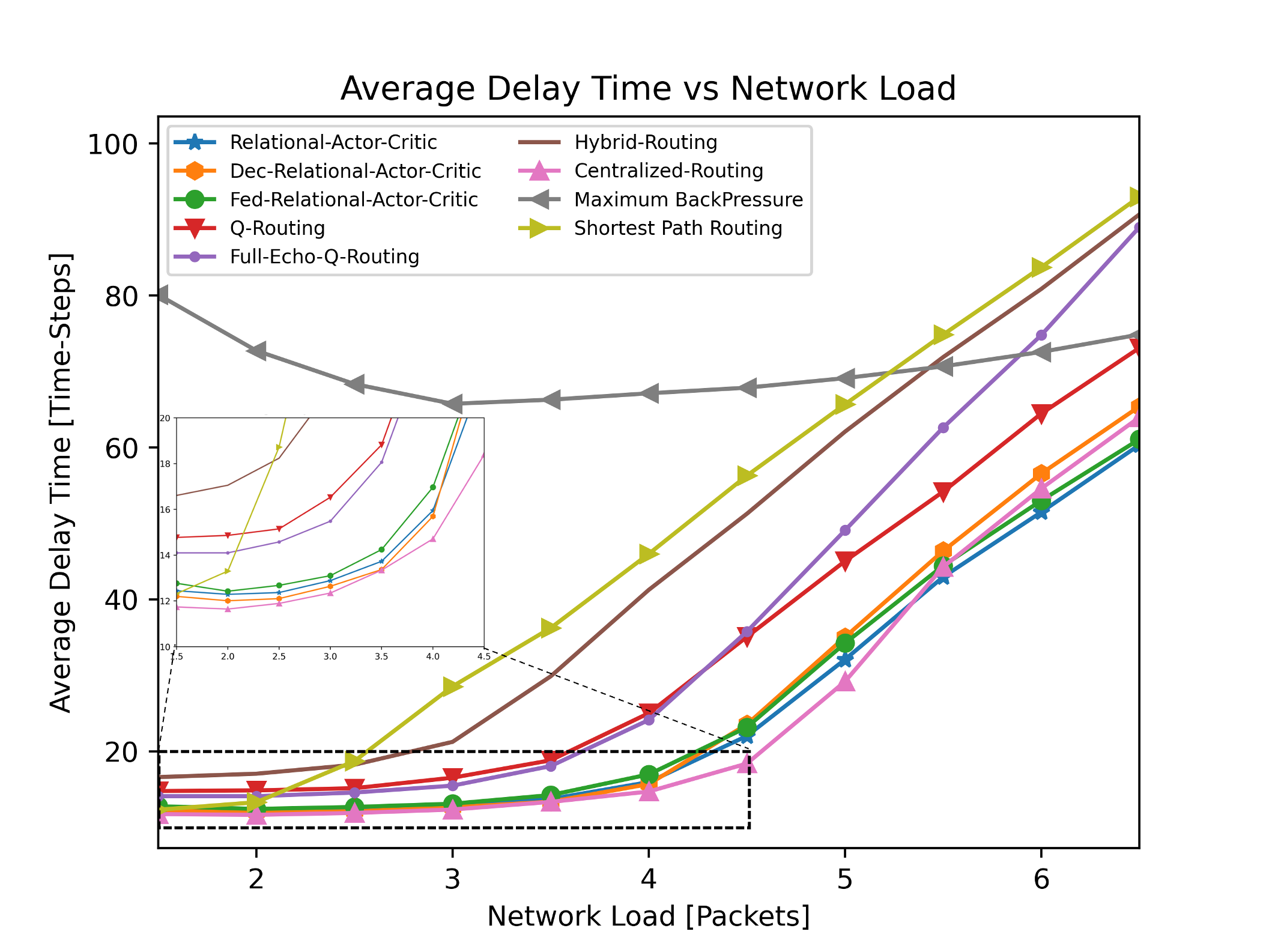}}\\
\caption{Static topology and dynamic topology - average delay performance for different routing algorithms under different loads.}
    \label{fig:load_sweep_average_delay}
\end{figure*}

\subsubsection{Static Topology}
This experiment evaluated the algorithm's performance under a static network topology and changing load. In Fig. \ref{fig:load_analysis_average_delay_static_topology}, we present the performance of the algorithms. Based on these results, we can conclude that all three versions of the Relational A2C algorithm outperformed other algorithms, despite acting in a decentralized manner. Moreover, these results also revealed that using exploration in conjunction with exploitation (Full-Echo Q-Routing) greatly improved performance at low to medium loads when compared with Q-Routing, whilst they achieved similar performances for higher loads. 
\par
As a further interesting phenomenon, we observe that Hybrid Routing improves performance when compared with traditional Q-Routing under low loads, but degrades performance under higher loads. We explain this phenomenon by stating that higher loads require better coordination among agents, as making an incorrect routing decision while the network is already congested will probably result in packet loss. Our conclusion is that since each agent learns an independent stochastic policy as part of the algorithm, it is less effective at handling non-stationary problems than traditional Q-Routing algorithms.
In our next experiment, we allow UEs to travel and change their base-station association, which will further increase the non-stationarity of the problem. 
\par
Additionally, as network load increases, we observed performance degradation of the Shortest-Path algorithm. Our explanation for this phenomenon is that queue delay does not have a significant influence on performance for low loads, but as the load increases, it becomes more significant. 

\subsubsection{Dynamic Topology}
This experiment evaluated the algorithm's performance under a dynamic network topology and changing load. Following this, based on the results illustrated in Fig. \ref{fig:load_sweep_average_delay} we can determine that although acting in a decentralized manner, all versions of Relational A2C's algorithm managed to achieve superior performance than the other algorithms for this scenario. Further, Hybrid Routing does not perform well in this non-stationary environment, as we observe further degradation in performance when compared to a static scenario, and at low load it is even less effective than Q-Routing.
\par
Moreover, when compared with the static scenario, Full-Echo Q-Routing exhibits performance degradation at medium to high loads. It performs worse in this range than the traditional Q-Routing. Our explanation for this phenomenon is that when using Full Echo Q-Routing, frequent topology changes may cause multiple modifications to the routing policy, which will result in increased instability and longer routing times as loads increase. In addition, this phenomenon coincides with the phenomena observed in \cite{boyan1994packet}, in which the authors applied those algorithms to a grid topology. 
\par
The significant performance gap between all versions of Relational A2C and the traditional hybrid algorithm supports our claim that by following our proposed algorithms, agents are able to cooperate more effectively and increase the stability of our routing policy, proving that agents that cooperate achieve better outcomes than agents who act selfishly. Thus, these results demonstrate that optimizing the joint goal and solving the coordination problem between agents significantly improves performance for both static and dynamic topologies. 
\par
When comparing the performance of static versus dynamic topologies as illustrated in Fig. \ref{fig:load_sweep_average_delay}, we observe that all tabular RL baselines suffer from performance degradation when users are permitted to move and change their base-station association. Based on these results, which indicate a significant difference between static and dynamic topologies, we propose the following experiment to examine the effect of UE movement on algorithm performance.

\subsubsection{Results from Examining the Influence of Dynamic Topology Changes}
In this experiment, we evaluated the algorithm's robustness to dynamic topology changes under constant loads. For this purpose, we scanned various UE speeds and measured the arrival ratio and average delay with the different algorithms. We evaluate the algorithms' performance for two possible loads, medium load ($\lambda=3$) and high load ($\lambda=5$), presented in Fig. \ref{fig:mobility_analysis_medium_load_average_delay} and Fig. \ref{fig:mobility_analysis_high_load_average_delay}, respectively.
\begin{figure*}[ht!]
  \subfloat[Average Delay For High Load $\lambda=5$. \label{fig:mobility_analysis_medium_load_average_delay}]{%
      \includegraphics[ width=0.49\textwidth]{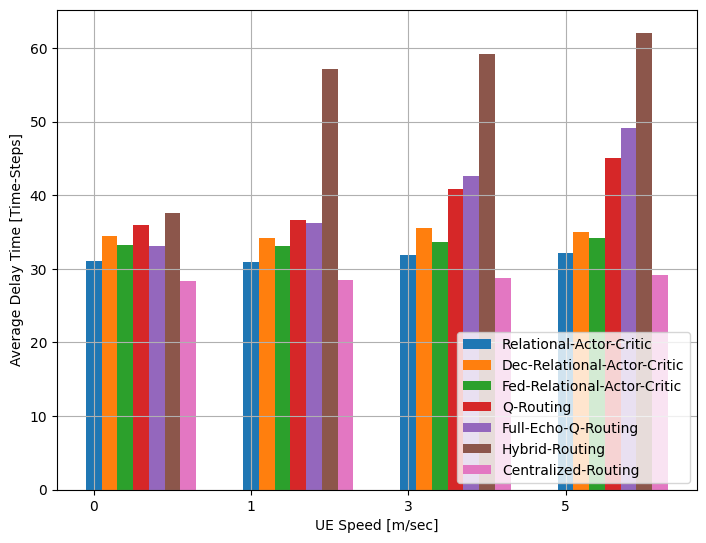}}
\hspace{\fill}
  \subfloat[Average Delay For Medium Load $\lambda=3$. \label{fig:mobility_analysis_high_load_average_delay}]{%
      \includegraphics[ width=0.49\textwidth]{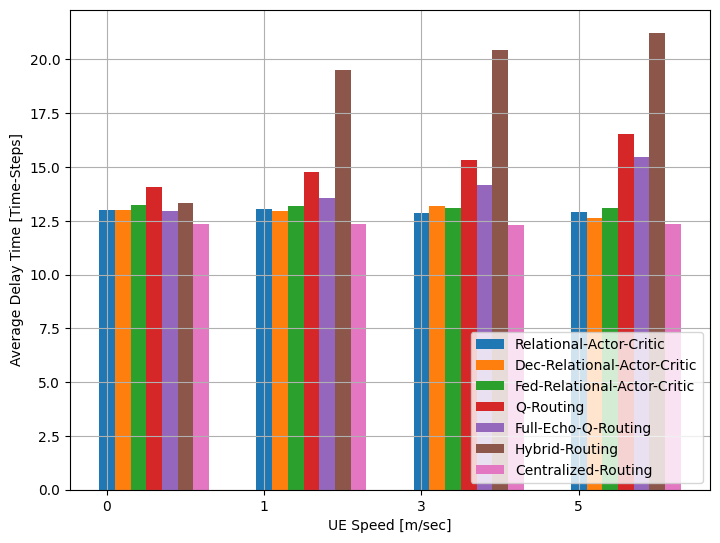}}\\
\caption{Average network delay for different routing algorithms under different UE speed.}
\label{fig:mobility_analysis_medium_load}
\end{figure*}
\par
Fig. \ref{fig:mobility_analysis_medium_load_average_delay} and Fig. \ref{fig:mobility_analysis_high_load_average_delay} demonstrate that increasing the UE speed will not have a significant impact on the Relational A2C algorithm's performance for varying network loads. Also, in medium loads, Full Echo Q-Routing achieves superior performance to Q-Routing, with both suffering from similar performance degradation as a result of their different UE speeds. Furthermore, we observe that the combination of higher load with increased speed results in more severe performance degradation for Full Echo Q-Routing compared to traditional Q-Routing. This observation is consistent with the results of our previous experiments. Moreover, Hybrid Routing shows the greatest degradation among all Q-Routing algorithms, further demonstrating that it is less able to cope with this non-stationary setting than traditional Q-Routing algorithms. Aside from this, we observe that all other RL-based algorithms suffer from performance degradation when UE speed is increased, which serves as an additional indication of our proposed approach's superiority.

\subsection{Experiment Results for Online Changes}
The purpose of this section is to examine how online changes impact algorithm performance in a variety of scenarios. Our study indicates that Relational A2C-based approaches are superior for handling such online changes when compared to other algorithms. As a next step, we will describe the online scenarios that we have studied. In the first case, we analyzed the algorithm's response to bursts of traffic. In the second case, we investigate how the algorithm responds to node failures and recovery situations. Also, in these experiments we capture those online changes by measuring the algorithm's performance using a sliding window with a length of 100 time slots.
\subsubsection{Experiment Results for Bursts of Traffic}
During this experiment we measure the algorithm's response to bursts of traffic. This was accomplished by evaluating the algorithm performance while changing the network load. Our experiment involves changing the network load rapidly from low ($\lambda=2$) to high ($\lambda=5$) and then back to low. The results presented here are an average of 10 measurements across five different topologies of the network. Furthermore, we conducted this experiment for both static and dynamic topologies, as decipted in Fig. \ref{fig:load_analysis_average_delay_static_topology} and Fig. \ref{fig:load_analysis_average_delay_dynamic_topology}, respectively. Further, due to the absence of further information offered by the arrival ratio, we present only the average delay metric measured through time for this experiement. 
\begin{figure*}[ht!]
  \subfloat[Average Delay For Static Topology. \label{fig:load_burst_average_delay_static_topology}]{%
      \includegraphics[ width=0.49\textwidth]{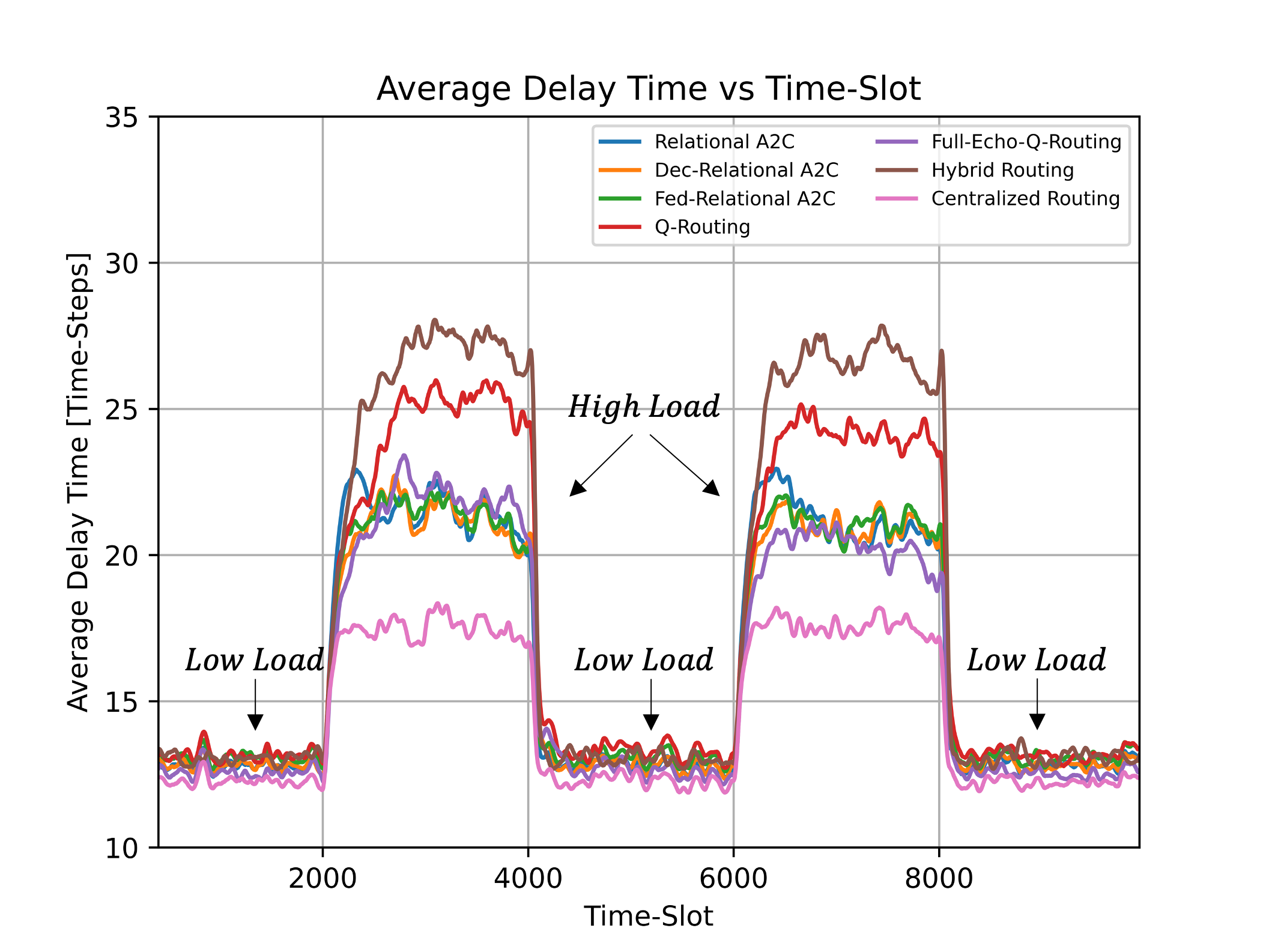}}
\hspace{\fill}
  \subfloat[Average Delay For Dynamic Topology. \label{fig:load_burst_average_delay_dynamic_topology}]{%
      \includegraphics[ width=0.49\textwidth]{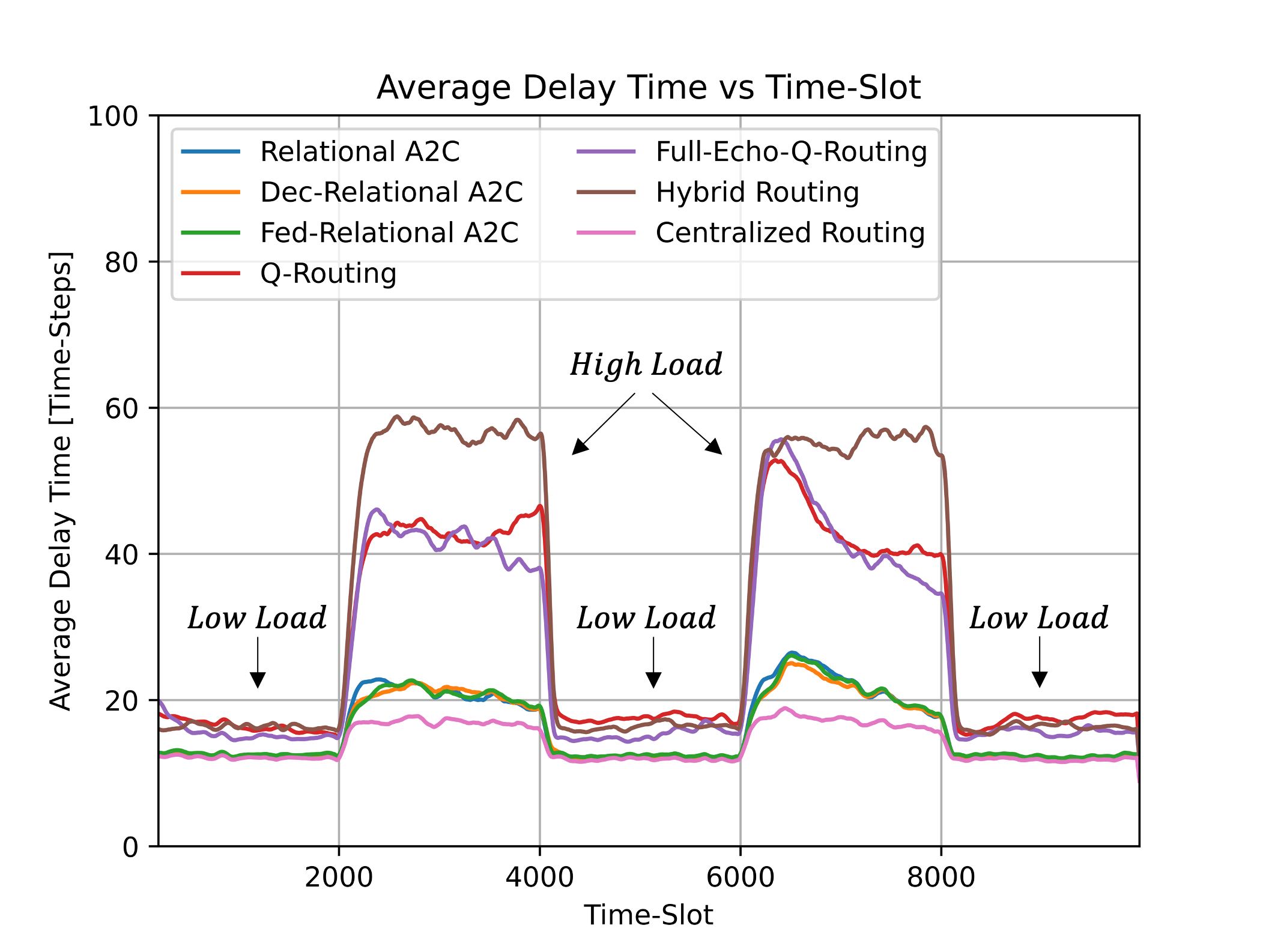}}\\
\caption{Static topology and dynamic topology - average delay performance for different routing algorithms under bursts of traffic. low load - $\lambda=2$, high load - $\lambda=5$.}
    \label{fig:load_burst_average_delay}
\end{figure*}
\par
In both static and dynamic scenarios, it is evident that the Relational A2C versions achieve superior performance than the traditional RL-algorithms.
Additionally, all algorithms exhibit similar reaction times when measuring the impact of changes in load from low to high and from high to low. When comparing between static and dynamic scenarios, we are able to see that the Relational A2C algorithms did not suffer from any performance degradation, while the tabular methods greatly suffer. 

Further, we observe that Full Echo Q-Routing is superior to traditional Q-Routing in static scenarios, while they achieve similar performance in dynamic scenarios. Additionally, Hybrid Routing suffers from the greatest degradation of all Q-Routing algorithms when working with high loads, emphasizing its inability to cope with such a non-stationary scenario.
Our conclusions are consistent with the results of our previous experiments, which further verifies the results.

\subsubsection{Experiment Results for Node Failure}
In this experiment, we evaluate how the algorithms respond to a scenario of node failure and recovery. To achieve this, the algorithm performance was evaluated while removing a random base station from the network for a specified period of time. Since we insert dynamic changes through node failures, we now consider only static topologies, in which users cannot move. The following results are based on an average of 30 consecutive experiments with random base-station failure in each experiment.
\par
This experiment was conducted for a high load as shown in Fig. \ref{fig:node_failure_performance_lambda_5}. To illustrate the arrival ratio metric in such an online setting, we measured the number of packets that dropped within our sliding window. From these results, it is evident that all algorithms are able to to recover from the node failure situation, although we observe that the average delay does not return to its initial value after the node has recovered. In our opinion, this is due to the fact that the other base stations' queues are already congested at this point. This results in a more congested network than before the node failure. In conjunction with the fact that the arrival rate process does not change over time, performance is degraded under higher loads.
\begin{figure*}[ht!]
  \subfloat[Average Delay Through Time. \label{fig:load_burst_average_delay_static_topology_lambda_5}]{%
      \includegraphics[ width=0.49\textwidth]{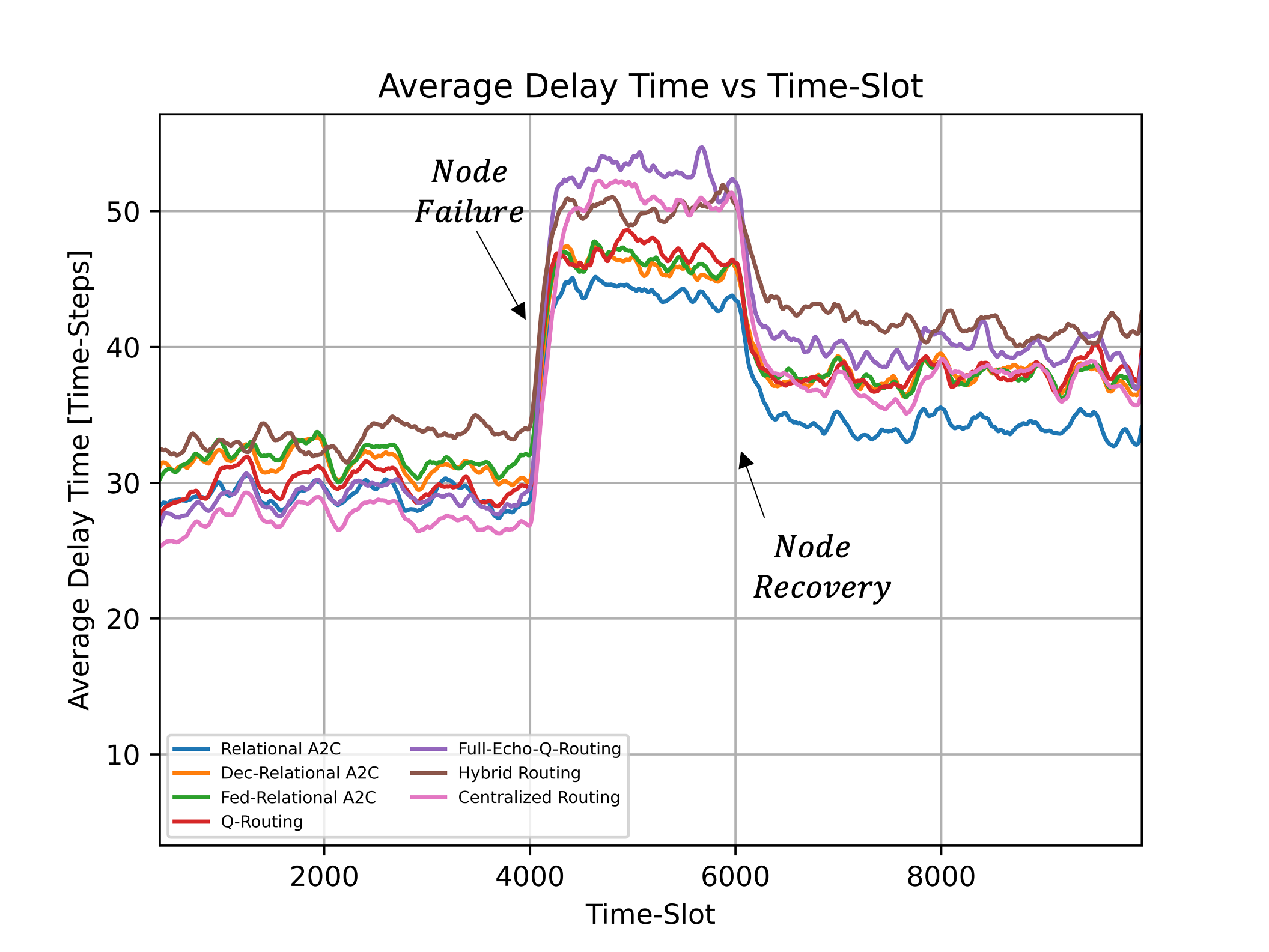}}
\hspace{\fill}
  \subfloat[Dropped Packets Through Time. \label{fig:load_burst_dropped_packets_static_lambda_5}]{%
      \includegraphics[ width=0.49\textwidth]{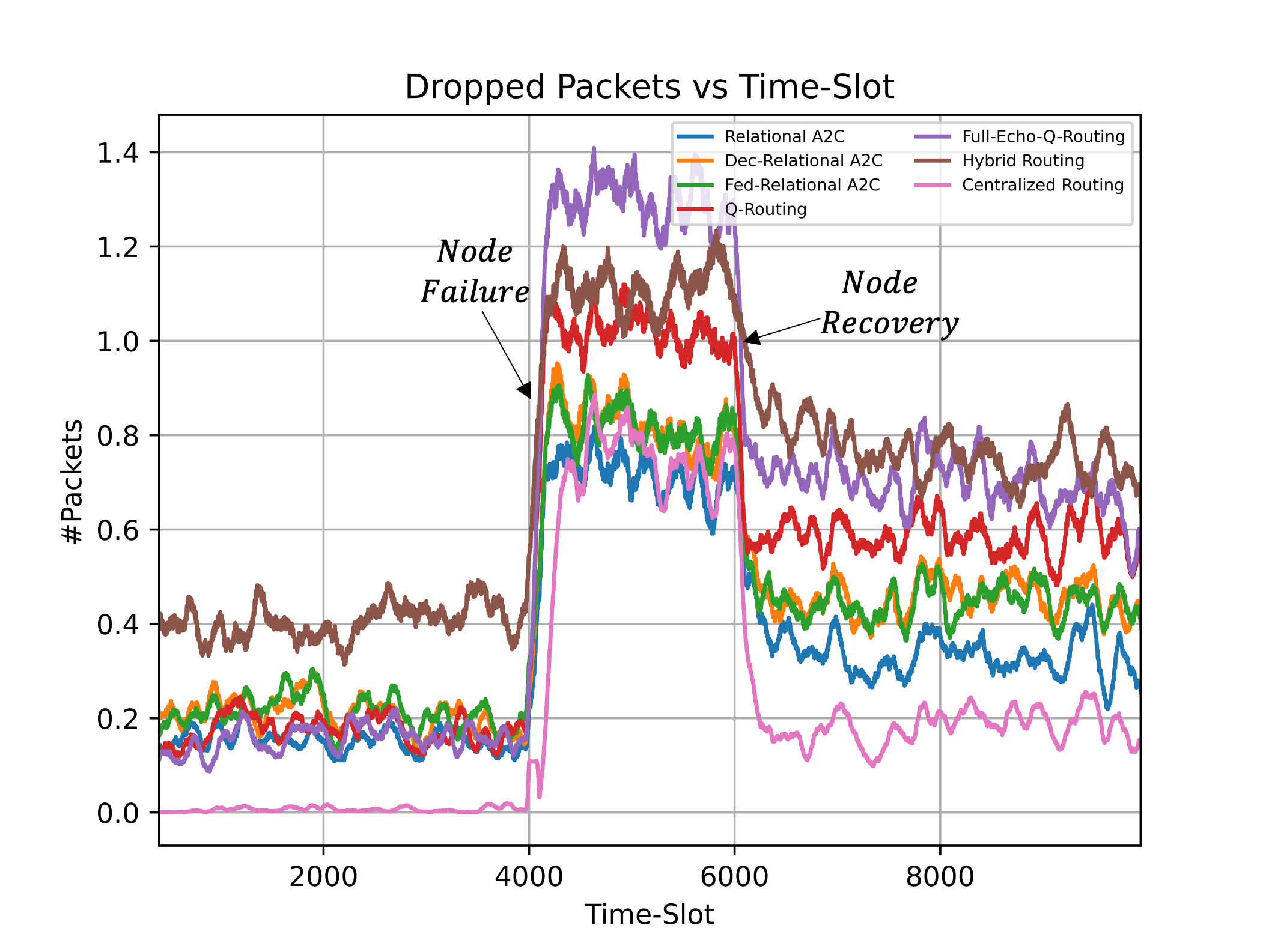}}\\
\caption{Performance for different routing algorithms in case of node failure for high load $\lambda=5$.}
    \label{fig:node_failure_performance_lambda_5}
\end{figure*}
\par
A further, but equally significant, conclusion we can draw from these measurements is that while Relational A2C has a lower delay than Centralized Routing in Fig. \ref{fig:load_burst_average_delay_static_topology_lambda_5}, its packet loss rate is higher in Fig. \ref{fig:load_burst_dropped_packets_static_lambda_5}. This demonstrates the importance of using both arrival rate and average delay in our analysis. 
\subsection{Algorithm Convergence}
In this section, we examine the algorithm's convergence time under a constant load, starting with a random initialization point. According to our results, all versions of our algorithm were able to converge to a similar solution, as opposed to the other baselines that suffer from performance degradation. Thus, we conclude that all Relational A2C-based algorithms were able to achieve a better routing solution than traditional algorithms. The results presented in Fig. \ref{fig:convergenceDelayMediumLoad} are the average of those measurements.
\begin{figure}
    \centering
    \includegraphics[width=0.49\textwidth]{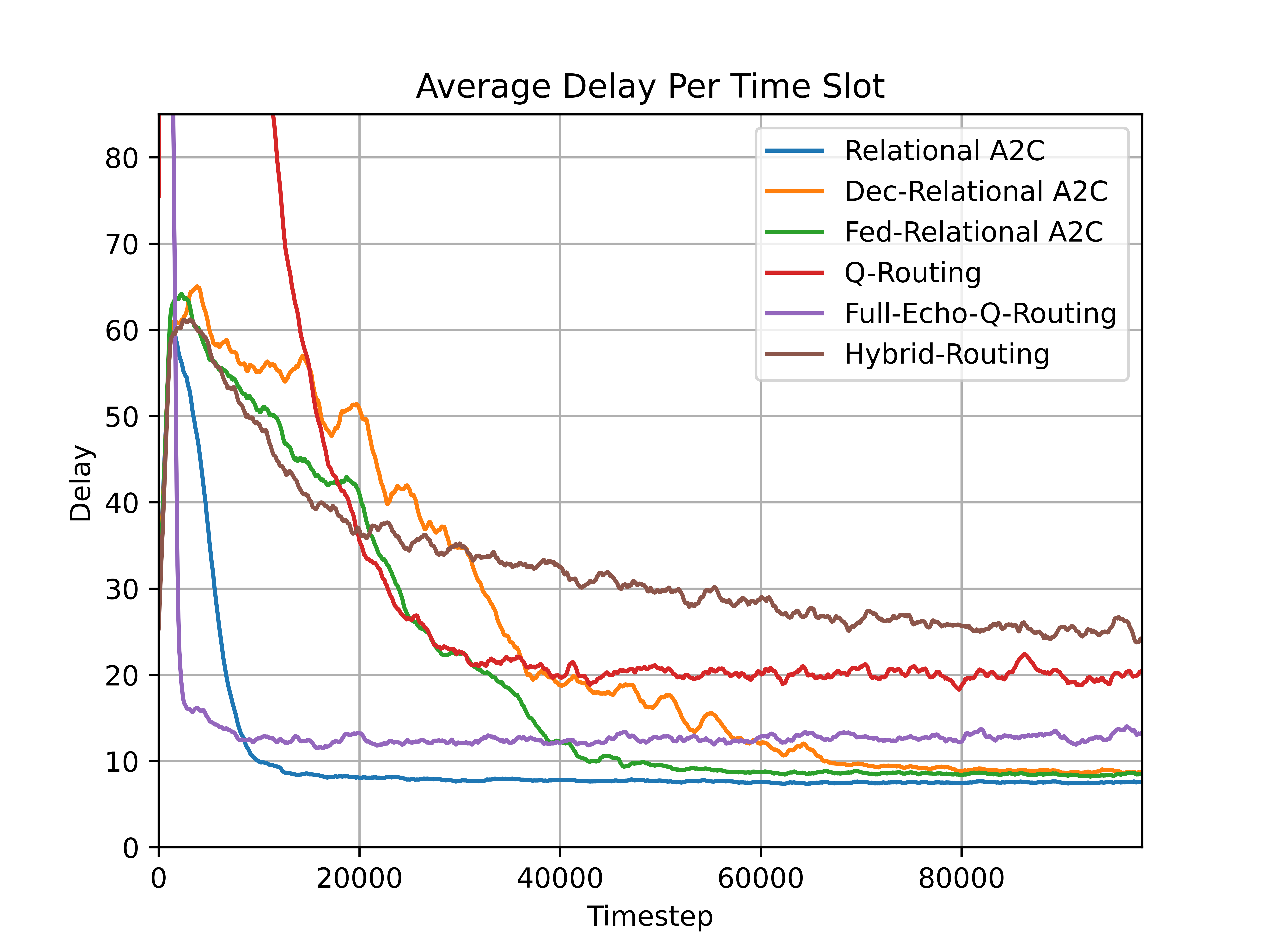}
    \caption{Illustration of the performance of different routing algorithms through their training procedures.}
    \label{fig:convergenceDelayMediumLoad}
\end{figure}
\par
Additionally, we found that centralized training achieved the fastest convergence among the different training paradigms of our proposed method. In fact, this convergence gap arises because the different agents can share their experience implicitly through the mutual updates of both the actor and the critic among all base stations. Furthermore, we conclude that Fed-Relational A2C performs better than Dec-Relational A2C in terms of convergence. In general, this difference can be understood intuitively, since federated learning approaches manage the trade-off between a fully decentralized training paradigm and a centralized training paradigm.
Another interesting phenomenon is that Full Echo Q-Routing appears to be the fastest algorithm that converges to a stable solution (an approximated point of equilibrium). This rapid convergence can be explained by the fact that Full Echo Q-Routing receives all of the rewards available to it, regardless of the chosen action. As a result, it is able to reduce the number of interactions required for convergence with the environment. 
\section{Conclusion}
In this paper, we examined the problem of routing in an IAB network in which multiple IABs nodes operate simultaneously within the same network in order to route multiple packets towards their destinations. A successful joint routing policy must be developed by the IABs in order to route packets effectively without causing network congestion. Due to physical limitations, the IAB is only able to exchange limited information with its neighbors, and thus does not know the current status of the entire network. This raises the interesting and unanswered question of which hops should be selected at each time step so that the network arrival ratio achieved by matching routing policies is maximal and the average packet delay is minimal. 
\par
To identify the joint routing policy that maximizes the network's arrival ratio while minimizing the average packet delay, we developed a novel Relational A2C  algorithm, which aims to determine the best joint routing strategy, based on observations collected by the IABs via online learning. To support decentralized training, we developed two different approaches which achieve similar performance to the centralized training approach. For various different scenarios, we compared the arrival ratio and average delay of the proposed Relational A2C  algorithms with those of six other algorithms. It was found that Relational A2C  performed better than baseline algorithms in all cases and achieved performance similar to that of a centralized approach.
Further, we demonstrate that network routing is crucial to reducing congestion and maximizing the utilization of network resources in IAB networks. These results clearly demonstrate the ability of Relational A2C based algorithms to learn near-centralized policies and the overall superiority of the proposed approach over existing methods.

\bibliographystyle{bib_style_ad_hoc_networks}
\bibliography{main}

\begin{thebibliography}{10}
\expandafter\ifx\csname url\endcsname\relax
  \def\url#1{\texttt{#1}}\fi
\expandafter\ifx\csname urlprefix\endcsname\relax\def\urlprefix{URL }\fi
\expandafter\ifx\csname href\endcsname\relax
  \def\href#1#2{#2} \def\path#1{#1}\fi

\bibitem{5G_spec}
3GPP, Nr; nr and ng-ran overall description; stage 2, Technical Specification
  (TS) 38~(300) (2019).

\bibitem{van1988beamforming}
B.~D. Van~Veen, K.~M. Buckley, Beamforming: A versatile approach to spatial
  filtering, IEEE assp magazine 5~(2) (1988) 4--24.

\bibitem{5G_IAB}
3GPP, Nr; study on integrated access and backhaul, Technical Specification (TS)
  38~(3874) (2019).

\bibitem{actor_critic}
V.~Konda, J.~Tsitsiklis, Actor-critic algorithms, Advances in neural
  information processing systems 12 (1999).

\bibitem{reiforce_algorithm}
R.~J. Williams, \href{https://doi.org/10.1007/BF00992696}{Simple statistical
  gradient-following algorithms for connectionist reinforcement learning},
  Mach. Learn. 8~(3–4) (1992) 229–256.
\newblock \href {https://doi.org/10.1007/BF00992696}
  {\path{doi:10.1007/BF00992696}}.
\newline\urlprefix\url{https://doi.org/10.1007/BF00992696}

\bibitem{td_learning}
G.~Tesauro, et~al., Temporal difference learning and td-gammon, Communications
  of the ACM 38~(3) (1995) 58--68.

\bibitem{boyan1994packet}
J.~A. Boyan, M.~L. Littman, Packet routing in dynamically changing networks: A
  reinforcement learning approach, in: Advances in neural information
  processing systems, 1994, pp. 671--678.

\bibitem{actor_critic_for_adaptive_routing_hybrid_method}
S.~Zeng, X.~Xu, Y.~Chen, Multi-agent reinforcement learning for adaptive
  routing: A hybrid method using eligibility traces, in: 2020 IEEE 16th
  International Conference on Control \& Automation (ICCA), 2020, pp.
  1332--1339.
\newblock \href {https://doi.org/10.1109/ICCA51439.2020.9264518}
  {\path{doi:10.1109/ICCA51439.2020.9264518}}.

\bibitem{bellman1958routing}
R.~Bellman, On a routing problem, Quarterly of applied mathematics 16~(1)
  (1958) 87--90.

\bibitem{backpressure_routing}
L.~Tassiulas, A.~Ephremides, Stability properties of constrained queueing
  systems and scheduling policies for maximum throughput in multihop radio
  networks, IEEE Transactions on Automatic Control 37~(12) (1992) 1936--1948.
\newblock \href {https://doi.org/10.1109/9.182479}
  {\path{doi:10.1109/9.182479}}.

\bibitem{akkaya2005survey}
K.~Akkaya, M.~Younis, A survey on routing protocols for wireless sensor
  networks, Ad hoc networks 3~(3) (2005) 325--349.

\bibitem{mammeri2019reinforcement}
Z.~Mammeri, Reinforcement learning based routing in networks: Review and
  classification of approaches, Ieee Access 7 (2019) 55916--55950.

\bibitem{johnson1996dynamic}
D.~B. Johnson, D.~A. Maltz, Dynamic source routing in ad hoc wireless networks,
  in: Mobile computing, Springer, 1996, pp. 153--181.

\bibitem{AODV}
C.~Perkins, E.~Royer, Ad-hoc on-demand distance vector routing, in: Proceedings
  WMCSA'99. Second IEEE Workshop on Mobile Computing Systems and Applications,
  1999, pp. 90--100.
\newblock \href {https://doi.org/10.1109/MCSA.1999.749281}
  {\path{doi:10.1109/MCSA.1999.749281}}.

\bibitem{DTNjain2004routing}
S.~Jain, K.~Fall, R.~Patra, Routing in a delay tolerant network, in:
  Proceedings of the 2004 conference on Applications, technologies,
  architectures, and protocols for computer communications, 2004, pp. 145--158.

\bibitem{sha2013multipath}
K.~Sha, J.~Gehlot, R.~Greve, Multipath routing techniques in wireless sensor
  networks: A survey, Wireless personal communications 70~(2) (2013) 807--829.

\bibitem{sutton2018reinforcement}
R.~S. Sutton, A.~G. Barto, Reinforcement learning: An introduction, MIT press,
  2018.

\bibitem{watkins1992q}
C.~J. Watkins, P.~Dayan, Q-learning, Machine learning 8~(3-4) (1992) 279--292.

\bibitem{choi1995predictive}
S.~Choi, D.-Y. Yeung, Predictive q-routing: A memory-based reinforcement
  learning approach to adaptive traffic control, Advances in Neural Information
  Processing Systems 8 (1995).

\bibitem{kumar1998confidence_q_routing}
S.~Kumar, R.~Miikkulainen, Confidence-based q-routing: An on-line adaptive
  network routing algorithm, in: Proceedings of Artificial Neural Networks in
  Engineering, Citeseer, 1998.

\bibitem{lecun2015deep}
Y.~LeCun, Y.~Bengio, G.~Hinton, Deep learning, nature 521~(7553) (2015)
  436--444.

\bibitem{mnih2015human}
V.~Mnih, K.~Kavukcuoglu, D.~Silver, A.~A. Rusu, J.~Veness, M.~G. Bellemare,
  A.~Graves, M.~Riedmiller, A.~K. Fidjeland, G.~Ostrovski, et~al., Human-level
  control through deep reinforcement learning, Nature 518~(7540) (2015)
  529--533.

\bibitem{stampa2017deep}
G.~Stampa, M.~Arias, D.~S{\'a}nchez-Charles, V.~Munt{\'e}s-Mulero, A.~Cabellos,
  A deep-reinforcement learning approach for software-defined networking
  routing optimization, arXiv preprint arXiv:1709.07080 (2017).

\bibitem{valadarsky2017learning}
A.~Valadarsky, M.~Schapira, D.~Shahaf, A.~Tamar, Learning to route, in:
  Proceedings of the 16th ACM workshop on hot topics in networks, 2017, pp.
  185--191.

\bibitem{Feat_Engineering_for_DRL}
J.~Suarez-Varela, A.~Mestres, J.~Yu, L.~Kuang, H.~Feng, P.~Barlet-Ros,
  A.~Cabellos-Aparicio, Feature engineering for deep reinforcement learning
  based routing, in: ICC 2019 - 2019 IEEE International Conference on
  Communications (ICC), 2019, pp. 1--6.
\newblock \href {https://doi.org/10.1109/ICC.2019.8761276}
  {\path{doi:10.1109/ICC.2019.8761276}}.

\bibitem{you2020toward}
X.~You, X.~Li, Y.~Xu, H.~Feng, J.~Zhao, H.~Yan, Toward packet routing with
  fully distributed multiagent deep reinforcement learning, IEEE Transactions
  on Systems, Man, and Cybernetics: Systems (2020).

\bibitem{Hierarchical_Deep_Double_Q_Routing}
R.~E. Ali, B.~Erman, E.~Baştuğ, B.~Cilli, Hierarchical deep double q-routing,
  in: ICC 2020 - 2020 IEEE International Conference on Communications (ICC),
  2020, pp. 1--7.
\newblock \href {https://doi.org/10.1109/ICC40277.2020.9149287}
  {\path{doi:10.1109/ICC40277.2020.9149287}}.

\bibitem{relational_drl}
V.~Manfredi, A.~P. Wolfe, B.~Wang, X.~Zhang, Relational deep reinforcement
  learning for routing in wireless networks, in: 2021 IEEE 22nd International
  Symposium on a World of Wireless, Mobile and Multimedia Networks (WoWMoM),
  2021, pp. 159--168.
\newblock \href {https://doi.org/10.1109/WoWMoM51794.2021.00029}
  {\path{doi:10.1109/WoWMoM51794.2021.00029}}.

\bibitem{kulkarni2016hierarchical}
T.~D. Kulkarni, K.~Narasimhan, A.~Saeedi, J.~Tenenbaum, Hierarchical deep
  reinforcement learning: Integrating temporal abstraction and intrinsic
  motivation, Advances in neural information processing systems 29 (2016)
  3675--3683.

\bibitem{Distributed_Path_Selection}
M.~Polese, M.~Giordani, A.~Roy, D.~Castor, M.~Zorzi, Distributed path selection
  strategies for integrated access and backhaul at mmwaves, in: 2018 IEEE
  Global Communications Conference (GLOBECOM), 2018, pp. 1--7.
\newblock \href {https://doi.org/10.1109/GLOCOM.2018.8647977}
  {\path{doi:10.1109/GLOCOM.2018.8647977}}.

\bibitem{simsek2020iab}
M.~Simsek, O.~Orhan, M.~Nassar, O.~Elibol, H.~Nikopour, Iab topology design: a
  graph embedding and deep reinforcement learning approach, IEEE Communications
  Letters 25~(2) (2020) 489--493.

\bibitem{puterman2014markov}
M.~L. Puterman, Markov decision processes: discrete stochastic dynamic
  programming, John Wiley \& Sons, 2014.

\bibitem{marl_selective_overview}
K.~Zhang, Z.~Yang, T.~Basar, \href{http://arxiv.org/abs/1911.10635}{Multi-agent
  reinforcement learning: {A} selective overview of theories and algorithms},
  CoRR abs/1911.10635 (2019).
\newblock \href {http://arxiv.org/abs/1911.10635} {\path{arXiv:1911.10635}}.
\newline\urlprefix\url{http://arxiv.org/abs/1911.10635}

\bibitem{boutilier1996planning}
C.~Boutilier, Planning, learning and coordination in multiagent decision
  processes, in: TARK, Vol.~96, Citeseer, 1996, pp. 195--210.

\bibitem{zhang2018fully}
K.~Zhang, Z.~Yang, H.~Liu, T.~Zhang, T.~Basar, Fully decentralized multi-agent
  reinforcement learning with networked agents, in: International Conference on
  Machine Learning, PMLR, 2018, pp. 5872--5881.

\bibitem{pareto_front_optimization}
P.~Ngatchou, A.~Zarei, A.~El-Sharkawi, Pareto multi objective optimization, in:
  Proceedings of the 13th International Conference on, Intelligent Systems
  Application to Power Systems, 2005, pp. 84--91.
\newblock \href {https://doi.org/10.1109/ISAP.2005.1599245}
  {\path{doi:10.1109/ISAP.2005.1599245}}.

\bibitem{fedAVGpaper}
B.~McMahan, E.~Moore, D.~Ramage, S.~Hampson, B.~A. y~Arcas,
  Communication-efficient learning of deep networks from decentralized data,
  in: Artificial intelligence and statistics, PMLR, 2017, pp. 1273--1282.

\bibitem{actor_critic_for_load_balancing}
T.~Mai, H.~Yao, Z.~Xiong, S.~Guo, D.~T. Niyato, Multi-agent actor-critic
  reinforcement learning based in-network load balance, in: GLOBECOM 2020 -
  2020 IEEE Global Communications Conference, 2020, pp. 1--6.
\newblock \href {https://doi.org/10.1109/GLOBECOM42002.2020.9322277}
  {\path{doi:10.1109/GLOBECOM42002.2020.9322277}}.

\bibitem{hu2003nash}
J.~Hu, M.~P. Wellman, Nash q-learning for general-sum stochastic games, Journal
  of machine learning research 4~(Nov) (2003) 1039--1069.

\bibitem{foerster2018counterfactual}
J.~Foerster, G.~Farquhar, T.~Afouras, N.~Nardelli, S.~Whiteson, Counterfactual
  multi-agent policy gradients, in: Proceedings of the AAAI conference on
  artificial intelligence, Vol.~32, 2018.

\bibitem{brockman2016openai}
G.~Brockman, V.~Cheung, L.~Pettersson, J.~Schneider, J.~Schulman, J.~Tang,
  W.~Zaremba, Openai gym, arXiv preprint arXiv:1606.01540 (2016).

\end{thebibliography}

\appendix
\newpage
\section{MDP Formulation}\label{appendix:MultiAgentMDPFormulation}
Under the assumption that each packet is an individual agent, this section introduces the proof that for the multiple packets scenario, we can formulate the problem of IAB network routing as an MDP with multiple agents. Table \ref{tab:multi_agent_mdp_table}. summarizes the Multi-Agent MDP formulation, while Lemmas \ref{lemma:mdp_form_2_reward}, \ref{lemma:mdp_form_2_transition} prove that it is indeed a Multi-Agent MDP. To formulate the problem as a Multi-Agent MDP we shall define the following tuple:  $<\mathcal{S},\mathcal{A} ,R,\mathcal{P},\gamma, N>$, where N represents the maximal number of packets (agents) our system supports.
\begin{table}[!h]
    \centering
    \begin{tabular}{|c|c|}
     \hline
         Type & Representation  \\
         \hline
         Action Space of agent i at time t & $\mathcal{A}^t_{i}=\bigg((\textbf{s}_{t,i,3}=Decision)\rightarrow\bigg\{a|\forall a\in\mathcal{N}, (n_{t,i},a)\in\mathcal{L}\bigg\}\bigg) \wedge \emptyset$ \\
         \hline
         Action Space at time t & $\mathcal{A}_t=\times_{i=0}^{N-1}\mathcal{A}_{t,i} $ \\
         \hline
         State & 
         $\textbf{s}_t=\bigcup_{i=0}^{N-1}\textbf{s}_{t,i}$,
         \\ 
         & where $\textbf{s}_{t,i}=\{n_{t,i}, ttl_{t,i}, ElapsedTime_{t,i},AgentMode_{t,i}\}$\\
         \hline
         Action & $\textbf{a}_t\in\mathcal{A}^t$ \\
         \hline
         Reward & $-\sum_{i=0}^{N-1}\bigg(ElapsedTime_{t,i}+d((n_{t,i},a_{t,i}))\bigg)\cdot\mathbbm{1}_{[s_{t,i,3} = 1]}$\\
         \hline
    \end{tabular}
    \caption{MDP formulation}
    \label{tab:multi_agent_mdp_table}
\end{table}
\begin{figure}[!h]
    \centering
    \includegraphics[width=\textwidth]{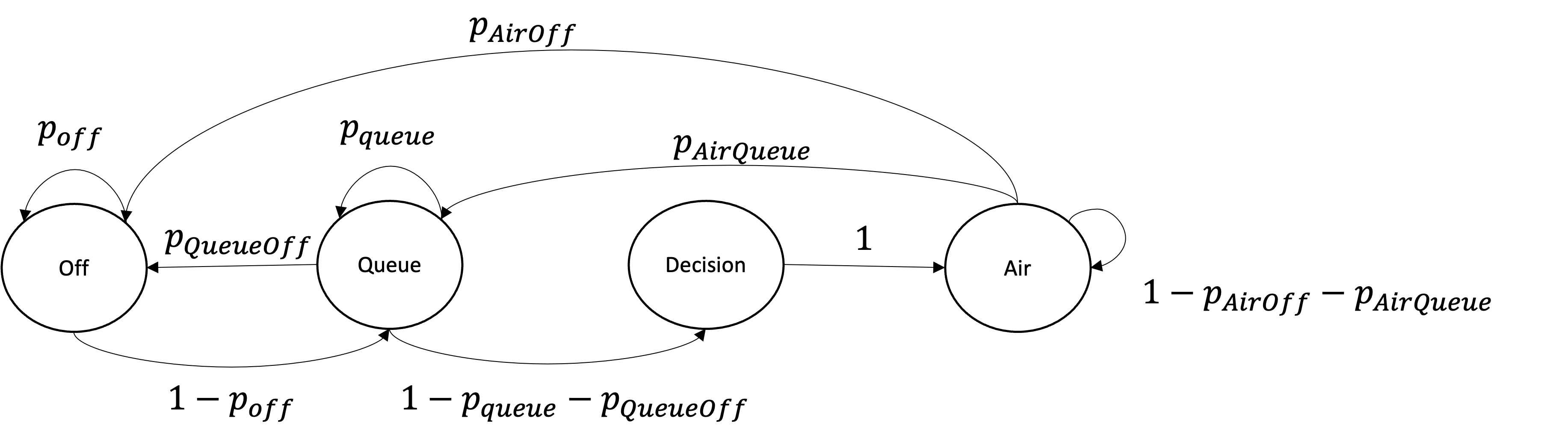}
    \caption{Different modes of each agent.}
    \label{fig:packet_mode}
\end{figure}
\par
Under this formulation, as illustrated in Fig.\ref{fig:packet_mode}, each agent has four available modes. The following explains the meaning of each mode, 
\begin{itemize}
    \item $Off$ - In this mode, the packet is not a part of the network at the moment, either because it has already been delivered to its destination, or suffered from a TTL expiration event, or because it is awaiting injection into the network.
    \item $Queue$ - In this mode, the packet is waiting in a specific base-station queue.
    \item $Decision$ - In this mode, the packet is currently required to make a next-hop decision.
    \item $Air$ - In this mode, the packet is transferred over the air from one node to another in the network.
\end{itemize}
We define our reward as follows:
\begin{equation}
    r_t(\textbf{s}_t,\textbf{a}_t) = -\sum_{i=0}^{N-1}\big(ElapsedTime_{t,i} + d((n_{t,i},a_{t,i}))\big)\cdot\mathbbm{1}_{[s_{t,i,3} = Decision]},
\end{equation}
specifically, $ElapsedTime_{t,i}$ represents the period of agent $i$ waiting at node $n_{t,i}$ queue before transmission and $d((n_{t,i},a_{t,i}))$ represents the transmission delay between node $n_{t,i}$ and node $a_{t,i}$ at time step $t\in\mathbb{N}$. The last term, $\mathbbm{1}_{[s_{t,i,3} == Decision]}$, represents if agent $i$ has permission to conduct a wireless hop or not. Next, we provide the full proof that this formulation is indeed a Multi-Agent MDP.
\begin{lemma}\label{lemma:mdp_form_2_reward}
 Given the history $\mathcal{H}_t$ the reward distribution depends only on $\textbf{s}_t$ and $\textbf{a}_t$. Thus, let $\Pr:\mathcal{S}\times\mathcal{A}\times\mathbb{Z}\rightarrow[0,1]$ be the agent's reward distribution.
\end{lemma}
\begin{proof}
    \begin{equation}
        \Pr(r_t|\mathcal{H}_t)\stackrel{\scriptstyle{(a)}}{=}\Pr(r_t|\textbf{s}^t,\textbf{a}^t,r^{t-1})\stackrel{\scriptstyle{(b)}}{=}\Pr(r_t|\textbf{s}_t,\textbf{a}_t).
    \end{equation}
    $(a)$ is true straight from the definition of history. Further, it follows that $(b)$ is true since the reward does not depend on the previous history once the action and state are determined. This claim is true since it is composed of $ElapsedTime_{t,i}$, which is a part of the current state, and $d((n_{t,i},a_{t,i}))$, which is a random variable that is dependent only on the current state and action.
\end{proof}
\begin{lemma}\label{lemma:mdp_form_2_transition}
     Given the history $\mathcal{H}_t$ the transition probability matrix $\Pr$ depends only on $\textbf{s}_t$ and $\textbf{a}_t$. Thus, let $\Pr:\mathcal{S}\times\mathcal{A}\times\mathcal{S}\rightarrow[0,1]$ be the transition probability matrix.
\end{lemma}
\begin{proof}
We begin by defining the following probabilities,
Whenever an agent reaches a TTL expiration event, its mode will transition to Off:
\begin{equation}
    \Pr\big(\textbf{s}_{t+1,i}=(\cdot,\cdot,\cdot,Off)|\textbf{s}_t,\textbf{a}_t, s_{t,i,1}=0\big) = 1.
\end{equation}
Next, we deal with the Air mode transition probabilities:
\begin{equation}
    \Pr\big(\textbf{s}_{t+1,i}=\big(a_t,ttl_{t,i}-1,d((n_{t,i},a_{t,i})),Air\big)|\textbf{s}_t,\textbf{a}_t,s_{t,i,3} = Decision\big) = 1.
\end{equation}
\begin{equation}
    \Pr\big(\textbf{s}_{t+1,i}=\big(n_{t,i},ttl_{t,i}-1,ElapsedTime_{t,i}-1,Air\big)|\textbf{s}_t,\textbf{a}_t,s_{t,i,3} = Air,s_{t,i,2} > 0\big) = 1.
\end{equation}
\begin{equation}
    \Pr\big(\textbf{s}_{t+1,i}=(n_{t,i},ttl_{t,i}-1,1,Queue)|\textbf{s}_t,\textbf{a}_t,s_{t,i,3} = Air,s_{t,i,2} = 0\big) = 1.
\end{equation}
Next, we explain why the current state combined with agent action provides sufficient information regarding queue transition probabilities, so the history condition does not matter. In the current state, the agent knows both the TTL values and the locations of all the agents. A prioritized queue, in combination with its state and action, enables the agent to determine which agents will transition from queue to decision mode next.
\par 
Last, the injection process of new agents in our system, which is based over sampling an i.i.d. random process, is independent at the agent history.
Thus, we can define our transition probabilities as follows:
\begin{align*}
    \Pr(\textbf{s}_{t+1}|\mathcal{H}_t) \stackrel{\scriptstyle{(a)}}{=}\Pr(\textbf{s}_{t+1}|\textbf{s}^t,\textbf{a}^t,r^{t-1})
    \stackrel{\scriptstyle{(b)}}{=}& 
    \Pr(\textbf{s}_{t+1}|\textbf{s}_t,\textbf{a}_t).
\end{align*}
Equality $(a)$ is from the definition of history. Equality $(b)$ is due to the fact that given that you are at state $\textbf{s}_t$ and you have chose action $\textbf{a}_t$, the previous states that you have visited have no influence over your next transition probabilities, as we have just shown for each agent mode.
\end{proof}

\section{Proof of Lemma \ref{lemma:gradient_relational}}\label{appendix:lemma1}
\begin{proof}
    Let $\tau(T) \triangleq (\textbf{s}_0, \textbf{o}_{0},\textbf{a}_0,r_{1},\ldots,\textbf{s}_{T-1},\textbf{o}_{T-1},\textbf{a}_{n},r_{T},\textbf{s}_{T},\textbf{o}_{T})$ represent a possible trajectory with length $T$. In addition, let $\Pr(\tau(T))\triangleq\Pr(\textbf{s}_0, \textbf{o}_{0},\textbf{a}_0,r_{1},\ldots,\textbf{s}_{T-1},\textbf{o}_{T-1},\textbf{a}_{T-1},r_{T},\textbf{s}_{T},\textbf{o}_{T})$,
by applying the chain rule and Lemmas \ref{lemma:mdp_form_2_reward}, \ref{lemma:mdp_form_2_transition}, we receive
$\Pr(\tau(T))=\Pr(\textbf{s}_{0})\Pi_{t=0}^{T-1}\big(\Pi_{n=0}^{N-1}\Pi(a_{t,n}|\textbf{o}_{t,n};\boldsymbol{\theta})\big)\Pr(\textbf{o}_{t}|\textbf{s}_{t})\Pr(\textbf{s}_{t+1},r_{t+1}|\textbf{s}_{t},\textbf{a}_{t})$.
Thus, we get
\begin{equation} \label{eq:prob_derivative}
    \nabla_{\boldsymbol{\theta}}\log\Pr(\tau(T))=\sum_{t=0}^{T-1}\sum_{n=0}^{N-1}\nabla_{\boldsymbol{\theta}}\log\Pi(a_{t,n}|\textbf{o}_{t,n};\boldsymbol{\theta}).
\end{equation}
Let us recall our goal, which is that we would like to solve the following optimization problem:
\[\Pi^\star = \text{arg}\max_{\Pi} J(\Pi) = \text{arg}\max_{\Pi}\underbrace{\mathbb{E}[\sum_{t=0}^{T-1}\gamma^{t}\cdot r_{t+1}|\Pi]}_{J(\Pi)}.\]
Following that, we can now calculate the gradient of the objective $J(\Pi)$ w.r.t. the policy parameters $\boldsymbol{\theta}$:

\begin{align*}
    \nabla_\theta J(\theta) & = \nabla_\theta\mathbb{E}\bigg[\frac{1}{N}\sum_{n=0}^{N-1}G_n\big]\stackrel{\scriptstyle{(a)}}{=} \sum_{\tau(T)}\nabla_{\boldsymbol{\theta}}\Pr(\tau(T))\cdot \frac{1}{N}\sum_{t=0}^{T-1}\gamma^{t}\cdot r_{t+1} \\
    &
    = \sum_{\tau(T)}\nabla_{\boldsymbol{\theta}}\Pr(\tau(T))\cdot\frac{\Pr(\tau(T))}{\Pr(\tau(T))}\cdot\sum_{t=0}^{T-1}\gamma^{t}\cdot r_{t+1} \\ 
    & = 
    \sum_{\tau(T)}\Pr(\tau(T))\nabla_{\boldsymbol{\theta}}\log\Pr(\tau(T))\cdot \sum_{t=0}^{T-1}\gamma^{t}\cdot r_{t+1} 
    \\ 
    &
    \stackrel{\scriptstyle{(b)}}{=}
    \sum_{\tau(T)}\Pr(\tau(T))\bigg(\sum_{t=0}^{T-1}\sum_{n'=0}^{N-1}\nabla_{\boldsymbol{\theta}}\log\Pi(a_{t,n'}|\textbf{o}_{t,n'};\boldsymbol{\theta})\bigg)\cdot\sum_{t'=0}^{T-1}\gamma^{t'}\cdot r_{t'+1} 
    \\
    &
    \stackrel{\scriptstyle{(c)}}{=}
    \mathbb{E}\bigg[\sum_{t=0}^{T-1}\bigg(\sum_{n'=0}^{N-1}\nabla_{\boldsymbol{\theta}}\log\Pi(a_{t,n'}|\textbf{o}_{t,n'};\boldsymbol{\theta})\bigg)\cdot \bigg(\sum_{t'=0}^{T-1}\gamma^{t'}\cdot r_{t'+1}\bigg)\bigg] 
    \\
    &
    \stackrel{\scriptstyle{(d)}}{=}
    \mathbb{E}\bigg[\sum_{t=0}^{T-1}\bigg(\sum_{n'\in\mathcal{I}_t}\nabla_{\boldsymbol{\theta}}\log\Pi(a_{t,n'}|\textbf{o}_{t,n'};\boldsymbol{\theta})\bigg)\cdot \bigg(\sum_{t'=t}^{T-1}\gamma^{t'}\cdot r_{t'+1}\bigg)\bigg] 
     \\
     &
     \stackrel{\scriptstyle{(e)}}{\propto}\mathbb{E}\bigg[\bigg(\sum_{n'\in\mathcal{I}_{t}}\nabla_{\boldsymbol{\theta}}\log\Pi(a_{t,n'}|\textbf{o}_{t,n'};\boldsymbol{\theta})\bigg)\cdot \bigg(\sum_{t'=t}^{T-1}\gamma^{t'}\cdot r_{t'+1}\bigg)\bigg]
    \\ 
    &
     \stackrel{\scriptstyle{(f)}}{=} \mathbb{E}_{\textbf{s}_t\sim\mu(s),\textbf{o}_t\sim\Pr(\cdot|\textbf{s}_t),\textbf{a}_t\sim\Pi(\cdot|\textbf{o}_t)}\bigg[\mathbb{E}\bigg[\bigg(\sum_{n'\in\mathcal{I}_{t}}\nabla_{\boldsymbol{\theta}}\log\Pi(a_{t,n'}|\textbf{o}_{t,n'};\boldsymbol{\theta})\bigg)\cdot \sum_{t'=t}^{T-1}\gamma^{t'}\cdot r_{t'+1})\bigg|\textbf{s}_t,\textbf{o}_t,\textbf{a}_t\bigg]\bigg]
     \\ 
    & 
    \stackrel{\scriptstyle{(g)}}{=}\mathbb{E}_{\textbf{s}_t\sim\mu(s),\textbf{o}_t\sim\Pr(\cdot|\textbf{s}_t),\textbf{a}_t\sim\Pi(\cdot|\textbf{o}_t)}\bigg[\sum_{n'\in\mathcal{I}_{t}}\nabla_{\boldsymbol{\theta}}\log\Pi(a_{t,n'}|\textbf{o}_{t,n'};\boldsymbol{\theta})\underbrace{\mathbb{E}\bigg[\sum_{t'=t}^{T-1}\gamma^{t'}\cdot r_{t'+1}\bigg|\textbf{s}_t,\textbf{a}_t\bigg]}_{Q_\Pi(\textbf{s}_t,\textbf{a}_t)}\bigg]
     \\ 
    &
    \stackrel{\scriptstyle{(h)}}{=} \mathbb{E}_{\textbf{s}_t\sim\mu(s),\textbf{o}_t\sim\Pr(\cdot|\textbf{s}_t),\textbf{a}_t\sim\Pi(\cdot|\textbf{o}_t)}\bigg[\sum_{n'\in\mathcal{I}_{t}}\nabla_{\boldsymbol{\theta}}\log\Pi(a_{t,n'}|\textbf{o}_{t,n'};\boldsymbol{\theta})\cdot Q_\Pi(\textbf{s}_t,\textbf{a}_t)\bigg]
     \\ 
    &
      \stackrel{\scriptstyle{(i)}}{=}\mathbb{E}_{\textbf{s}_t\sim\mu(s),\textbf{o}_t\sim\Pr(\cdot|\textbf{s}_t),\textbf{a}_t\sim\Pi(\cdot|\textbf{o}_t)}\bigg[\bigg(\sum_{n'\in\mathcal{I}_{t}}\nabla_{\boldsymbol{\theta}}\log\Pi(a_{t,n'}|\textbf{o}_{t,n'};\boldsymbol{\theta})\cdot \big(Q_\Pi(\textbf{s}_t, \textbf{a}_t)-\underbrace{\frac{1}{|\mathcal{I}_t|}\sum_{n\in\mathcal{I}_{t}}\hat{V}_{\Pi}(\textbf{o}_{t,n};\textbf{w})}_{b(\textbf{o}_t)}\big)\bigg)\bigg].
    \end{align*}

Equality $(a)$ follows from the definition of expectation combined with the linearity of both expectation and derivative (under the assumption that the trajectories set is independent of the policy parameters). Equality $(b)$ follows from Eq. \ref{eq:prob_derivative}, and equality $(c)$ follows from the definition of expectation. Next, equality $(d)$ holds since the policy of agents that are not active at a certain time slot $t\in(0,\ldots,T-1)$ is independent of the policy parameters, and therefore, the corresponding derivative is equal to zero. The proportion relation $(e)$ holds under the assumption of an MDP with a stationary distribution, which we denote with $\mu(s), s\in\mathcal{S}$. 
Equality $(f)$ holds according to the law of total expectation. 
Equality $(g)$ holds since the derivative of the agent's policy is deterministic when considering the current state, observation and action, and therefore, it can be moved out of the inner expectation. Equality $(h)$ follows from the definition of action-value function and from the fact that given the current state and action, the current observation is irrelevant to the expected return. Last, equality $(i)$ is due to the fact that reducing a baseline function that is independent of the policy actions does not introduce any bias under the expectation \cite[Ch.~13]{sutton2018reinforcement}.
\end{proof}
\section{Proof of Lemma \ref{lemma:dec_gradient_approx}}\label{appendix:lemma2}
\begin{proof}
    Let $\tau(T) \triangleq (\textbf{s}_0, \textbf{o}_{0},\textbf{a}_0,r_{1},\ldots,\textbf{s}_{T-1},\textbf{o}_{T-1},\textbf{a}_{n},r_{T},\textbf{s}_{T},\textbf{o}_{T})$ represent a possible trajectory with length $T$. In addition, we denote the joint parameters $\boldsymbol{\theta}$ as follows: $\boldsymbol{\theta}\triangleq\bigcup_{k=0}^{K-1} \boldsymbol{\theta}_k$. Next, let $\Pr(\tau(T))\triangleq\Pr(\textbf{s}_0, \textbf{o}_{0},\textbf{a}_0,r_{1},\ldots,\textbf{s}_{T-1},\textbf{o}_{T-1},\textbf{a}_{T-1},r_{T},\textbf{s}_{T},\textbf{o}_{T})$,
by applying the chain rule and Lemmas \ref{lemma:mdp_form_2_reward}, \ref{lemma:mdp_form_2_transition}; thus we receive
$\Pr(\tau(T))=\Pr(\textbf{s}_{0})\Pi_{t=0}^{T-1}\big(\Pi_{n=0}^{N-1}\Pi(a_{t,n}|\textbf{o}_{t,n};\boldsymbol{\theta})\big)\Pr(\textbf{o}_{t}|\textbf{s}_{t})\Pr(\textbf{s}_{t+1},r_{t+1}|\textbf{s}_{t},\textbf{a}_{t}).$
Thus, we get
\begin{equation} \label{eq:prob_derivative_dec}
\nabla_{\boldsymbol{\theta}_k}\log\Pr(\tau(T))=\sum_{t=0}^{T-1}\sum_{n=0}^{N-1}\nabla_{\boldsymbol{\theta}_k}\log\Pi(a_{t,n}|\textbf{o}_{t,n};\boldsymbol{\theta}) \text{ for } \forall k\in(0,\ldots,K-1).    
\end{equation}
Next, we derive the gradient of the joint objective $J(\Pi)$ w.r.t. $\boldsymbol{\theta}_k$.
\begin{align*}
     \nabla_{\boldsymbol{\theta}_k} J(\theta) & = \nabla_\theta\mathbb{E}\bigg[\frac{1}{N}\sum_{n=0}^{N-1}G_n\big]\stackrel{\scriptstyle{(a)}}{=} \sum_{\tau(T)}\nabla_{\boldsymbol{\theta}_k}\Pr(\tau(T))\cdot \sum_{t=0}^{T-1}\gamma^{t}\cdot r_{t+1}  \\
    &
    = \sum_{\tau(T)}\nabla_{\boldsymbol{\theta}_k}\Pr(\tau(T))\cdot\frac{\Pr(\tau(T))}{\Pr(\tau(T))}\cdot\sum_{t=0}^{T-1}\gamma^{t}\cdot r_{t+1}  \\ 
    & = 
    \sum_{\tau(T)}\Pr(\tau(T))\nabla_{\boldsymbol{\theta}_k}\log\Pr(\tau(T))\cdot \sum_{t=0}^{T-1}\gamma^{t}\cdot r_{t+1}  
    \\ 
    &
    \stackrel{\scriptstyle{(b)}}{=} \sum_{\tau(T)}\Pr(\tau(T))\bigg(\sum_{t=0}^{T-1}\sum_{n'=0}^{N-1}\nabla_{\boldsymbol{\theta}_k}\log\Pi(a_{t,n'}|\textbf{o}_{t,n'};\boldsymbol{\theta})\bigg)\sum_{t'=0}^{T-1}\gamma^{t'}\cdot r_{t'+1} 
    \\
    &
    \stackrel{\scriptstyle{(c)}}{=}
    \mathbb{E}\bigg[\sum_{t=0}^{T-1}\bigg(\sum_{n'=0}^{N-1}\nabla_{\boldsymbol{\theta}_k}\log\Pi(a_{t,n'}|\textbf{o}_{t,n'};\boldsymbol{\theta})\bigg)\cdot \bigg(\sum_{t'=0}^{T-1}\gamma^{t'}\cdot r_{t'+1}\bigg)\bigg] 
    \\
    &
    \stackrel{\scriptstyle{(d)}}{=}
    \mathbb{E}\bigg[\sum_{t=0}^{T-1}\bigg(\sum_{n'\in\mathcal{I}_{k,t}}\nabla_{\boldsymbol{\theta}_k}\log\Pi(a_{t,n'}|\textbf{o}_{t,n'};\boldsymbol{\theta}_k)\bigg)\cdot \bigg(\sum_{t'=t}^{T-1}\gamma^{t'}\cdot r_{t'+1}\bigg)\bigg] 
     \\
    &
     \stackrel{\scriptstyle{(e)}}{\propto}\mathbb{E}\bigg[\bigg(\sum_{n'\in\mathcal{I}_{k,t}}\nabla_{\boldsymbol{\theta}_k}\log\Pi(a_{t,n'}|\textbf{o}_{t,n'};\boldsymbol{\theta}_k)\bigg)\cdot \bigg(\sum_{t'=t}^{T-1}\gamma^{t'}\cdot r_{t'+1}\bigg)\bigg]
    \\ 
    &
     \stackrel{\scriptstyle{(f)}}{=} \mathbb{E}_{\textbf{s}_t\sim\mu(s),\textbf{o}_t\sim\Pr(\cdot|\textbf{s}_t),\textbf{a}_t\sim\Pi(\cdot|\textbf{o}_t)}\bigg[\mathbb{E}\bigg[\bigg(\sum_{n'\in\mathcal{I}_{k,t}}\nabla_{\boldsymbol{\theta}_k}\log\Pi(a_{t,n'}|\textbf{o}_{t,n'};\boldsymbol{\theta}_k)\bigg)\cdot \sum_{t'=t}^{T-1}\gamma^{t'}\cdot r_{t'+1})\bigg|\textbf{s}_t,\textbf{o}_t,\textbf{a}_t\bigg]\bigg]
     \\ 
      & 
    \stackrel{\scriptstyle{(g)}}{=}\mathbb{E}_{\textbf{s}_t\sim\mu(s),\textbf{o}_t\sim\Pr(\cdot|\textbf{s}_t),\textbf{a}_t\sim\Pi(\cdot|\textbf{o}_t)}\bigg[\sum_{n'\in\mathcal{I}_{k,t}}\nabla_{\boldsymbol{\theta}_k}\log\Pi(a_{t,n'}|\textbf{o}_{t,n'};\boldsymbol{\theta}_k)\underbrace{\mathbb{E}\bigg[\sum_{t'=t}^{T-1}\gamma^{t'}\cdot r_{t'+1}\bigg|\textbf{s}_t,\textbf{a}_t\bigg]}_{Q_\Pi(\textbf{s}_t,\textbf{a}_t)}\bigg]
     \\ 
      &
    \stackrel{\scriptstyle{(h)}}{=} \mathbb{E}_{\textbf{s}_t\sim\mu(s),\textbf{o}_t\sim\Pr(\cdot|\textbf{s}_t),\textbf{a}_t\sim\Pi(\cdot|\textbf{o}_t)}\bigg[\sum_{n'\in\mathcal{I}_{k,t}}\nabla_{\boldsymbol{\theta}_k}\log\Pi(a_{t,n'}|\textbf{o}_{t,n'};\boldsymbol{\theta}_k)\cdot Q_\Pi(\textbf{s}_t,\textbf{a}_t)\bigg]
     \end{align*}
     \begin{align*}
     \\ 
    & \stackrel{\scriptstyle{(i)}}{=} \mathbb{E}\bigg[\bigg(\sum_{n'\in\mathcal{I}_{k,t}}\nabla_{\boldsymbol{\theta}_k}\log\Pi(a_{t,n'}|\textbf{o}_{t,n'};\boldsymbol{\theta}_k)\cdot \big(Q_\Pi(\textbf{s}_t, \textbf{a}_t)-\underbrace{\frac{1}{|\mathcal{I}_{k,t}|}\sum_{n\in\mathcal{I}_{k,t}}\hat{V}_{\Pi}(\textbf{o}_{t,n};\textbf{w})}_{b(\textbf{o}_t)}\big)\bigg)\bigg].
    \end{align*}

Equality $(a)$ follows from the definition of expectation combined with the linearity of both the expectation and derivative (Under the assumption that the trajectories set is independent of the policy parameters). Equality $(b)$ follows from Eq. \ref{eq:prob_derivative_dec}, and equality $(c)$ follows from the definition of expectation. Next, equality $(d)$ holds since the policy of agents that are active in a certain time slot $t\in(0,\ldots,T-1)$ is dependent on the policy parameters of the $k^{th}$ base-station only if they are acting within these base stations, and therefore, the remaining agent's derivative is equal to zero (the ones that are active outside of base-station $k$ or the ones that are not active at all). The proportion relation $(e)$ holds under the assumption of a MDP with a stationary distribution, which we denote with $\mu(s), s\in\mathcal{S}$. 
The equality $(f)$ holds according to the law of total expectation. 
Equality $(g)$ holds since the derivative of the agent's policy is deterministic when considering the current state, observation and action, and therefore, it can be moved out of the inner expectation. Equality $(h)$ follows from the definition of action-value function and from the fact that given the current state and action, the current observation is irrelevant to the expected return. Last, equality $(i)$ is due to the fact that reducing a baseline function that is independent with policy actions does not introduce any bias under the expectation \cite[Ch.~13]{sutton2018reinforcement}.
\end{proof}
\end{document}